\documentclass[a4paper,american,reqno]{amsart}

\usepackage{babel}
\usepackage[utf8]{inputenc}
\usepackage[T1]{fontenc}
\usepackage[binary-units=true]{siunitx}
\usepackage{csquotes}
\usepackage{todonotes}
\usepackage{booktabs}
\usepackage{amssymb}
\usepackage{enumitem}
\usepackage{dcolumn}
\usepackage{algorithm,algpseudocode}
\usepackage[style = authoryear-comp,
            maxbibnames = 100,
            maxcitenames = 2,
            giveninits = true,
            uniquename = init,
            isbn = false,
            backend = bibtex]{biblatex}
\usepackage[colorlinks,
            citecolor=blue,
            urlcolor=blue,
            linkcolor=blue]{hyperref}
\usepackage{cleveref}

\makeatletter
\patchcmd{\@settitle}{\uppercasenonmath\@title}{\scshape\large}{}{}
\patchcmd{\@setauthors}{\MakeUppercase}{\scshape\normalsize}{}{}
\makeatother

\vfuzz \hfuzz
\makeatletter
\@namedef{subjclassname@2020}{%
  \textup{2020} Mathematics Subject Classification}
\makeatother

\newcommand{\jscom}{\todo[author=JS, color=yellow!50, size=\small]}
\newcommand{\jscomil}{\todo[inline, author=JS, color=yellow!50, size=\small, caption={}]}

\newenvironment{proofClaim}[1][]{\ifthenelse{\equal{#1}{}}{\begin{proof}}{\begin{proof}[#1]}}{\end{proof}}

\DeclareMathOperator*{\argmin}{arg\,min}
\DeclareMathOperator*{\argmax}{arg\,max}
\newcommand{\inte}{\mathrm{int}}
\newcommand{\con}{\mathrm{con}}
\newcommand{\inter}[1]{\mathrm{int}\left(#1\right)}
\newcommand{\cl}[1]{\mathrm{cl}\left(#1\right)}

\newcommand{\NI}{\Psi}
\newcommand{\proxy}{\eta}
\newcommand{\proxyagg}{\eta_{\mathrm{agg}}}
\newcommand{\VBR}{\Phi}

\newcommand{\BR}{\mathrm{BR}}

\newcommand{\abs}[1]{\lvert#1\rvert}
\newcommand{\extray}{r}
\newcommand{\IntBound}{\rho}
\newcommand{\xint}{x^{\inte}}
\newcommand{\xcon}{x^{\con}}

 \newcommand{\txint}{\tilde{x}^{\inte}}
 \newcommand{\bxint}{\bar{x}^{\inte}}
 
\NewDocumentCommand{\A}{O{}O{}O{}}{A^{#1,#2}_{#3}}
\RenewDocumentCommand{\b}{O{}O{}}{b^{i}_{#2}}

\newcommand{\R}{\mathbb{R}}
\newcommand{\Z}{\mathbb{Z}}

\newcommand{\FeasRt}{F_t}
\NewDocumentCommand{\freeset}{O{}}{S^{#1}}
\newcommand{\powset}{\mathcal{P}}
\newcommand{\BRcom}{\mathcal{BR}}
\newcommand{\Fa}{\mathrm{Fa}}
\newcommand{\Opt}{\mathrm{Opt}}
\newcommand{\Nes}{\mathcal{E}}
\newcommand{\Nespi}{\mathcal{E}^\pi}
\newcommand{\Corner}{K}

\newcommand{\Nbas}{\mathcal{N}}
\NewDocumentCommand{\FeasStrats}{O{}}{W_{#1}}
\NewDocumentCommand{\FeasStratsInt}{O{}}{W_{#1}^\inte}
\NewDocumentCommand{\FeasStratsCon}{O{}}{W_{#1}^\con}
\NewDocumentCommand{\FeasStratsRel}{O{}}{\hat{W}_{#1}}
\newcommand{\Rall}{\prod_{j \in N}\R^{k_j+l_j}}
\newcommand{\Rallmini}{\prod_{j \neq i}\R^{k_j+l_j}}

\newcommand{\AddThm}[2]{
	\newtheorem{#1}[thmcntr]{#2}
	\AddToHook{env/#1/begin}{\crefalias{thmcntr}{#1}}
}

\theoremstyle{definition}
\AddThm{definition}{Definition}
\AddThm{example}{Example}
\AddThm{proposition}{Proposition}
\AddThm{lemma}{Lemma}
\AddThm{theorem}{Theorem}
\AddThm{conjecture}{Conjecture}
\AddThm{question}{Question}
\AddThm{corollary}{Corollary}
\AddThm{claim}{Claim}
\AddThm{assumption}{Assumption}
\AddThm{remark}{Remark}
\AddThm{invariant}{Invariant}
\AddThm{subclaim}{Subclaim}

\Crefname{lemma}{Lemma}{Lemmas}

\newcommand{\wrt}{w.r.t.\ }

\NewDocumentCommand{\NeInt}{O{\ }}{TODO{#1}}
\NewDocumentCommand{\intcon}{O{\ }}{\textcolor{red}{integer-parameterized constrained#1}}
\NewDocumentCommand{\bases}{O{\ }}{\textcolor{red}{good name#1}}

\newcommand{\defset}[3][\defsep]{\set{#2#1#3}}
\newcommand{\Defset}[3][\defsep]{\Set{#2#1#3}}
\newcommand{\set}[1]{\{#1\}}
\newcommand{\Set}[1]{\left\{#1\right\}}
\newcommand{\st}{\text{s.t.}}
\newcommand{\define}{\mathrel{{\mathop:}{=}}}

\newcommand{\nonNegInts}{\mathbb{Z}_{\geq 0}}

\newcommand{\rev}[1]{\begingroup#1\endgroup}
\newcommand{\revised}[1]{\begingroup#1\endgroup}

\newcommand{\revsecond}[1]{#1}

\newcolumntype{d}[1]{D{.}{.}{#1}}

\newlist{thmparts}{enumerate}{1}
\setlist[thmparts]{
	label=(\alph*),
	leftmargin=15pt,
}

\usepackage{xstring}
\usepackage{etoolbox}
\makeatletter
\apptocmd{\cref@getref}{\xdef\@lastusedlabel{#1}}{}{error}
\makeatother

\Crefformat{thmpart}{%
	\StrCount{\@lastusedlabel}{:}[\LastColonPos]%
	\StrBefore[\LastColonPos]{\@lastusedlabel}{:}[\ThmLabel]%
	\nameCref{\ThmLabel}~#2\ref*{\ThmLabel}#1#3%
}
\Crefmultiformat{thmpart}{%
	\StrCount{\@lastusedlabel}{:}[\LastColonPos]%
	\StrBefore[\LastColonPos]{\@lastusedlabel}{:}[\ThmLabel]%
	\nameCref{\ThmLabel}~#2\ref*{\ThmLabel}#1#3%
}{ and~#2#1#3}{, #2#1#3}{ and~#2#1#3}

\Crefname{assumption}{Assumption}{Assumption}
\Crefname{proposition}{Proposition}{Propositions}
\Crefname{theorem}{Theorem}{Theorems}
\Crefname{lemma}{Lemma}{Lemmas}
\Crefname{claim}{Claim}{Claims}

\bibliography{branch-and-cut-for-ipgs}

\begin{document}

\title[Branch-and-Cut for Mixed-Integer GNEPs]%
{Branch-and-Cut for Mixed-Integer\\
  Nash Equilibrium Problems}

\author[A. Duguet, T. Harks, M. Schmidt, J. Schwarz]%
{Aloïs Duguet, Tobias Harks, Martin Schmidt, Julian Schwarz}

\address[A. Duguet, M. Schmidt]{%
  Trier University,
  Department of Mathematics,
  Universitätsring 15,
  54296 Trier,
  Germany}
\email{duguet@uni-trier.de}
\email{martin.schmidt@uni-trier.de}

\address[T. Harks, J. Schwarz]{%
  University of Passau,
  Faculty of Computer Science and Mathematics,
  94032 Passau,
  Germany}
\email{tobias.harks@uni-passau.de}
\email{julian.schwarz@uni-passau.de}

\date{\today}

\begin{abstract}
  We consider mixed-integer Nash equilibrium problems in which
each player solves a mixed-integer optimization problem
parameterized in the rivals' strategies. We distinguish
standard Nash equilibrium problems (NEPs), where the
parameterization acts only on the cost functions, and
generalized NEPs (GNEPs), where the strategy spaces may also
depend on the rivals' strategies. We introduce a
branch-and-cut (B\&C) algorithm for such games that, upon
termination, either computes a pure Nash equilibrium or
decides that none exists. The main idea is a bilevel
reformulation based on the Nikaido--Isoda function. Using
bilevel-optimization techniques, we obtain a tractable
relaxation of this reformulation and embed it into a B\&C
framework. Depending on the setting, we derive sufficient
conditions for the existence of suitable cuts and finite
termination. For GNEPs, we adapt intersection cuts (ICs)
from bilevel and mixed-integer linear optimization. Such
cuts exist under assumptions that hold, e.g., for
pure-integer GNEPs with decoupled concave objectives and
linear coupling constraints. Since ICs may become
arbitrarily shallow, we complement them with a fallback
strategy that fixes recurring integer profiles and ensures
finite termination in a relevant setting. For NEPs, suitable
cuts always exist via best-response inequalities, and our
method terminates finitely whenever the set of best-response
sets is finite. This holds, e.g., if (i) the cost functions
are concave in the players' own continuous strategies and
the strategy sets are polyhedral, or (ii) the cost functions
only depend on the player's own strategy and the rivals'
integer strategy components. Numerical results on various
NEPs and GNEPs show that our tailored B\&C method outperforms
directly applying the state-of-the-art bilevel solver MibS
on the tested integer and generalized knapsack instances.

\end{abstract}

\keywords{Generalized Nash equilibrium problems,
Mixed-integer games,
Bilevel optimization,
Branch-and-cut,
Cutting planes%
%
%
}
\subjclass[2020]{90C11, 
90C47, 
90C57, 
91-08
%
}

\maketitle

\section{Introduction}
\label{sec:introduction}

\rev{Nash equilibrium problems have applications in
  various domains ranging from market games in economics
  \parencite{Debreu54}, communication networks \parencite{Kelly98},
  transport systems \parencite{Beckmann56}, to electricity markets
  \parencite{Anderson13}.
  While the computation of Nash equilibria has been
  extensively studied for decades, the focus has traditionally been on
  continuous and convex settings and only over the recent decades, attention has
  also turned towards nonconvex games.
  From an application perspective, particularly important are classes of
  games with integer constraints on the player's decision
  variables \parencite{Koeppe11}.
  Prominent examples for this type of restriction include
  integer-splittable congestion games \parencite{Rosenthal73b}, market
  games with discrete goods \parencite{Kelso82} or energy markets with
  discrete quantities \parencite{guo2021copositive}.
  A fundamental question for these games is the existence of pure Nash
  equilibria (NE)---which arguably is the most convincing solution concept.
  Due to the non-convexity of strategy spaces stemming from \emph{integer
  constraints}, however, the existence of such equilibria is not
  guaranteed and in this regard, algorithms that compute equilibria or
  decide their non-existence are important.}

\subsection{Our Contribution}

\revised{We consider \emph{mixed-integer} Nash equilibrium problems
in which each player solves a mixed-integer optimization problem
parameterized in the rivals’ strategies.
Within this broad class of games, we distinguish between standard Nash
equilibrium problems (NEPs), where the parameterization acts only
on the players' cost functions and generalized Nash equilibrium
problems (GNEPs), where, additionally, the strategy spaces of the
players may depend on the rivals' strategies.}

As our main contribution, we introduce the first branch-and-cut (B\&C)
framework \revised{for such mixed-integer games.}  Our framework is
based on exploiting a connection between GNEPs
 and bilevel optimization. We use the well-known reformulation of a GNEP
using the Nikaido--Isoda function and formulate a corresponding
bilevel model. We then adapt the machinery of bilevel optimization
to obtain a tractable relaxation of the problem, the so-called
\emph{high-point relaxation} (HPR), which we further relax by dropping
the integrality constraints, arriving at the \emph{continuous
  high-point relaxation} (C-HPR).
Our  B\&C method then works  as follows.
\begin{enumerate}[label=(\roman*)]
\item\label{enum:high-point} We  solve the node problem, initialized
  with (C-HPR) at the root node in our search tree. In case that the
  node problem is infeasible or if it admits an optimal objective
  value strictly larger than 0, the node does not contain an NE and is
  pruned.
\item Branching Phase:
  Given an optimal solution (with non-positive value), we
  create new nodes as usual by branching on fractional integer
  variables. Once we obtain an integer-feasible node solution, we
  check if it is actually  an NE, and, if so, we stop.
\item Cutting Phase: Otherwise, we shrink the feasible set via
  newly introduced \emph{non-NE-cuts}, i.e., via cuts given by
  inequality constraints that are violated by the current
  integer-feasible node solution but satisfied by any
  NE. After adding the cut, we jump back to
  the first step and solve the augmented (C-HPR).
\end{enumerate}

\rev{Assuming that we can always derive a non-NE-cut,
  we prove the correctness of this method in \Cref{thm:correctness_algBC},
  i.e., we show that if the algorithm terminates, it either outputs
  an equilibrium or gives a certificate of non-existence.}
The key challenge in the above framework is to generate appropriate
non-NE-cuts. \rev{For their construction in the GNEP setting,}
we draw inspiration from the work
of~\textcite{fischetti2018use}, who consider mixed-integer bilevel
optimization problems and apply intersection cuts (ICs),
originally proposed by~\textcite{balas1971intersection} for integer
programming.
For our framework, we construct non-NE-cuts via intersection cuts, which
require (i) a cone pointed at the node solution containing the set of
all equilibria of the respective sub-tree of the B\&C tree and (ii) a
convex set containing the node solution in
its interior but no Nash equilibrium.
\revised{We construct a set satisfying (ii) under the assumption that the cost
  and constraint functions are convex in the rivals’ strategies and that
  the constraint function only attains integral values.}
Moreover, we provide sufficient conditions for the existence of a cone
fulfilling (i), e.g., that \rev{both} the sum of all players' cost functions
\revised{and the constraints are} linear.
\revsecond{We note that, even for single-level mixed-integer optimization problems,
  B\&C algorithms relying on ICs may fail to terminate in finite time because
  the ICs can become arbitrarily shallow.
  In fact, a general finite termination result for
  \revsecond{our B\&C method using only ICs and selecting
  	 vertex node solutions}
  is impossible:
  We showed in \parencite{duguet2026mixedintegergamesadmitrational}
  that even simple GNEPs with rational input data may
  admit only irrational equilibria and we provide an adapted example
  in the appendix (see~\Cref{ex: IrrationalLinSC})
  that fits to our framework with the same properties.
  In this situation, \revsecond{our aforementioned B\&C method}
  cannot terminate in finite time, \revsecond{since every IC is rationally described and hence every
    vertex node solution of this example
    is provably rational.}
  However, for the case of the players'
  constraint sets only depending on the integer profile and only finitely
  many best-response sets occur for each such profile,
  we nevertheless obtain finite termination
  via a fallback strategy (\Cref{sec:fallback}). If an integer profile reappears
  too often, we branch until all integer variables are fixed to it and process
  the resulting node with the best-response cuts which we describe in the following.
}

\rev{For NEPs, it turns out that we can formulate non-NE-cuts
  without the use of ICs but based on best-response inequalities.
  \revsecond{These cuts are globally valid for NEPs without additional
  assumptions while we use them locally in Section~\ref{sec:fallback}  for
  GNEPs after fixing an integer profile.}
  Moreover, they are, roughly speaking, at
  least as strong as ICs for NEPs; see \Cref{lem: ICCutStandNEP}.
  Note that similar best-response cuts have been used before
  by~\textcite{dragotto2023zero} for a cutting-plane algorithm in
  the pure integer  NEP setting. We embed these cuts in a general B\&C framework,
  which also applies to the class of mixed-integer strategy spaces.
  For these best-response cuts, we show that the B\&C method
  terminates in finite time if the set of best-response sets of
  all players are finite (\Cref{thm: FinitelyManyCutsGeneral}).
  This condition is then shown to be satisfied for the  important special cases of
  (i) the players' cost functions being concave in their own continuous
  strategies \revsecond{with polyhedral strategy sets} and
  (ii) the players' cost functions only depending on their
  own strategy and the rivals' integer strategy components (\Cref{lem:
    FinitelyManyBRSetsConcave,lem: FinitelyManyBRSetsInteger}).
  Note that class~(i) encompasses, in particular, the mixed extension
  of any \revsecond{finite} NEP.\footnote{\revsecond{A finite Nash
      equilibrium problem consists of finitely many players having
      each a finite set of strategies. The \emph{mixed extension} is
      now the strategic game with the same player set but strategy
      spaces are given by the respective probability simplex over the
      finitely many strategies. \revsecond{Costs are the expected costs
      of the original game under independent randomization and are
      linear in each player's own probabilities.} The famous theorem
      of~\textcite{nash1950equilibrium} shows existence of mixed Nash
      equilibria, which are defined as ordinary (pure) Nash
      equilibria of the constructed mixed-extension.
  Note that for a succinctly represented game, explicitly listing all
  pure strategies may require exponential space and we do not claim an efficient
  conversion or a computational advantage of our method over specialized
  mixed-equilibrium algorithms.}
  }
  Therefore, our framework is also applicable to the
  computation of mixed NE.}

Finally, we implemented our B\&C method and provide numerical
results in \Cref{sec:numerical-results} \wrt \rev{four} different games:
knapsack games (NEP), generalized knapsack games (GNEP),
implementation games (GNEP), and \rev{NEPs with quadratic objectives from \textcite{schwarze2023branch}.}

\subsection{Related Work}

Continuous and convex GNEPs in which the players' strategy spaces and
cost functions are convex have been studied intensively in terms of
existence theory and numerical algorithms. In this regard,
we refer to the survey articles of \textcite{facchinei2010generalized}
and \textcite{FISCHER2014} for an overview of the general theory and
focus in the following mainly on the nonconvex case.

Previous works on mixed-integer GNEPs and the computation of (pure) NE
are very sparse with the exceptions of
\textcite{Sagratella2017AlgorithmsFG,Sagratella19} and \textcite{Harks24}.
In the former paper \parencite{Sagratella2017AlgorithmsFG}, the author considers
generalized linear potential games, i.e., he assumes jointly constrained GNEPs with
linear coupling constraints and decoupled linear cost functions.
For this class, he shows that a Gauß--Seidel best-response (BR)
algorithm  approximates equilibria within arbitrary precision  using a finite
number of steps in this setting.
A similar best-response method is also used by
\textcite{Fabiani_Grammatico:2020} in an applied context of modeling
multi-vehicle automated driving.
Further proximal-like best-response methods for mixed-integer GNEPs
are studied by \textcite{Fabiani_et_al:2022}.
In \textcite{Sagratella19}, the author considers mixed-integer
GNEPs in which each player's strategy set depends on the other
players' strategies via a linear constraint in her own strategy and
the other players' integrally constrained strategies.
Sagratella proposes a branch-and-bound method based on a general merit
function and on the Nikaido--Isoda function. Moreover, he presents a
branch-and-prune method exploiting the idea of dominance of strategies for
pruning. Under the assumption of strict monotonicity of the
objective's derivative, he provides sufficient conditions for
strategies to be dominated.
Similarly, finite termination for the branch-and-prune method is shown
under the additional assumption that the gradient of the costs in the
continuous variables is strictly monotone.
In his paper, Sagratella also uses the term ``cut''. However, this
``cut'' is not used as a valid inequality to strengthen any relaxation
but rather in a pruning step of the branching algorithm.
\textcite{Harks24} consider general nonconvex GNEPs with
nonconvex strategy spaces and nonconvex cost functions.
Based on a convexification technique using the Nikaido--Isoda function,
they provide a novel characterization of equilibria by associating
with every GNEP instance a set of convexified instances.
They then introduce the class of quasi-linear models,
where a convexified instance exists in which for fixed strategies of
the opponent players, the cost function of every player
is linear and the respective strategy space is polyhedral.
For this class of games, they show that the convexification reduces
the GNEP to a standard (nonlinear) optimization problem.
They provide a numerical study regarding the computation of
equilibria for three classes of quasi-linear GNEPs related to integral
network flows and discrete market equilibria. Their general approach
is limited in the sense that it relies on deriving a convexification,
which itself is known to be computationally demanding. In contrast,
the method presented in this paper circumvents this step and offers a
direct computational approach.

While mixed-integer GNEPs remain relatively unexplored, there is
growing interest in the study of general formulations of nonconvex
 Nash equilibrium problems. A notable development in this
direction is the emergence of integer programming games (IPGs), first
introduced in the seminal work by \textcite{Koeppe11}, where each player's optimization
problem involves minimizing a continuous function over a fixed
polyhedral feasible set with partially integral variables. Since then,
IPGs are the subject of intensive research; see, e.g., \textcite{carvalho2018existence,
  carvalho2022computing,carvalho2023integer,Carvalho_et_al:2026,cronert2022equilibrium,
  PiaFM17,KleerS17,guo2021copositive,kirst2022branch,dragotto2023zero}.
Among these works, \textcite{kirst2022branch,dragotto2023zero} are the
ones most closely related to ours
due to their focus on \emph{pure} NE as well as their adaptation
of classic techniques from mixed-integer
linear optimization such as branching or cutting.
\textcite{kirst2022branch} propose a branch-and-bound algorithm for
computing the set of all approximate equilibria within a specified
error tolerance for IPGs with box-constraints. By exploiting this
special structure, their approach relies on rules that
identify and eliminate regions of the feasible space that cannot
contain any equilibria.
\textcite{dragotto2023zero} address the computation, enumeration, and
selection of Nash equilibria in IPGs with purely integral strategy
spaces using a \revised{cutting-plane algorithm. In contrast to our
framework, their method solves integer problems at
intermediate steps and thus is a multi-tree method
while our method is a single-tree branch-and-cut algorithm that is not
restricted to pure integer games.}

\revsecond{In recent work, \textcite{Hao_et_al:2026} develop a
generalized zero-regret algorithm for finite integer games with
shared constraints using conditional equilibrium inequalities.
They independently reuse our generalized-knapsack benchmark set;
we discuss their reported performance in
Section~\ref{subsubsec:gnep-knapsack-results}.
Our framework also covers continuous strategic variables, including
the prices in the implementation games.}

Besides \textcite{kirst2022branch,dragotto2023zero},
\textcite{Carvalho_et_al:2026,carvalho2022computing} also adapt
classic techniques from mixed-integer linear optimization.
\revised{
However, their focus is on computing mixed NE
to circumvent the difficulties arising from the non-convexity of IPGs,
which allows for quite
different algorithm designs as
explained below.
Note that while we focus on pure NE, mixed NE
correspond to pure NE of the mixed extension of a game
and can therefore also be handled within our framework.}
\textcite{Carvalho_et_al:2026} show that, under mild assumptions,
the computation of a mixed NE reduces to computing a pure NE
in a related convexified game, where each player’s
strategy space is replaced by its convex hull.
Consequently, the complexity of finding a mixed NE in
a nonconvex game shifts from dealing with non-convexities
to constructing the convexified game.
Similar to \textcite{Harks24}, this convexification step
is itself computationally demanding.
Rather than separating the construction of the convexified game
from the equilibrium search, \textcite{Carvalho_et_al:2026} propose the
cut-and-play algorithm, which integrates both tasks in one iterative
process.
The algorithm starts with a polyhedral outer approximation of the
convexified game, i.e., a game where players’ strategy spaces are
polyhedral supersets of their convexified counterparts.
It then iteratively refines the approximation by attempting to
compute a mixed NE of it. If no equilibrium exists or the computed
equilibrium is infeasible for the original convexified game,
the strategy spaces are selectively refined.
In contrast to the cuts we employ, these refinements are
player-specific and preserve the invariant that the convex hulls
of the original strategy sets remain feasible.

In contrast to \textcite{Carvalho_et_al:2026}, \textcite{carvalho2022computing}
propose an inner approximation algorithm through
their sampled generation method (SGM).
Inspired by column generation in mixed-integer linear optimization,
their algorithm begins with a finite subgame of the original IPG and
iteratively computes mixed NE. In each iteration, each player’s
strategy in the current mixed NE is compared to their best response
in the full game. If the latter yields a higher payoff,
it is added to the subgame. This process iteratively refines the approximation
of the original game, continuing until a mixed NE of the full IPG is found.
The authors emphasize that the effectiveness of SGM critically relies
on the efficient computation of mixed NE in the subgames.
Since these subgames are finite normal-form games, the existence of
mixed NE is guaranteed, and they can be computed relatively efficiently.
In contrast, the existence and computation of pure NE in nonconvex IPGs
are significantly more challenging, making an adaptation of SGM
to the pure-strategy setting appear impractical.

Besides IPGs, let us also mention the class of nonconvex NEPs
considered by \textcite{Sagratella16,schwarze2023branch}
in which each player's strategy set is given by a convex restriction
function combined with integrality constraints for all strategy
components.
Under the additional assumption that the cost function of
every player is convex in her own strategy and continuous in the
complete strategy profile, \textcite{Sagratella16} proposes a
branching method to compute the entire set of NE.
\rev{The method from the latter paper has later been used for the
  special case of so-called 2-groups partitionable games in
  \textcite{Sagratella:2017}.}
By enhancing this method via a pruning procedure,
\textcite{schwarze2023branch} are able to drop the assumption of
players' cost functions to be convex. Similar to
\textcite{Sagratella19}, their proposed pruning procedure exploits
dominance arguments based on the derivative of the cost function.

\revised{Branch-and-prune techniques (B\&P) as discussed by
  \textcite{Sagratella16,Sagratella19,schwarze2023branch} are closely
  related to branch-and-cut methods. We compare our B\&C method to the
  B\&P method by \textcite{schwarze2023branch} in our numerical study
  and brief\/ly outline in the following the differences between the two methods.
  First of all, the goal of the two methods is different:
  B\&P looks for all NE while our B\&C stops after finding
  the first NE.
  Second, B\&P relies on pruning nodes of the branching tree,
  while B\&C improves the relaxation of the node problems via cuts.
  Third, the typical assumptions are not the same.
  B\&P methods are essentially based on some kind of strict monotonicity
  of the costs, while our B\&C requires other assumptions for ensuring
  correctness and finite termination, such as linear constraints,
  concavity of the sum of costs, etc.\ so that ICs can be derived.
  Fourth and finally, the computational costs depend on the core
  mechanism. For the B\&P of \textcite{schwarze2023branch}, checking
  the dominance criterion is typically cheap. In our B\&C,
  deriving an IC for GNEPs is more expensive, but deriving a
  best-response cut for NEPs is again cheap.}

\revsecond{Let us mention that several papers in the preceding literature
  like \textcite{cronert2022equilibrium} mention the case of
  mixed-integer games explicitly as important but
  difficult. \textcite{schwarze2023branch} mentioned the
  generalization of integer games towards GNEPs as important and open.
  To the best of our knowledge, this paper is the first general
  computational approach to address these problems in one framework.}

Finally, let us also mention the most relevant works from
mixed-integer bilevel optimization, which serves as the basis for our
solution approach.
This field started with the seminal paper by
\textcite{Moore_Bard:1990}, in which the first branch-and-bound method
for mixed-integer bilevel optimization is discussed.
The first branch-and-cut method is due to
\textcite{DeNegre_Ralphs:2009}, which motivated many more recent
contributions in this field such as
\textcite{Fischetti-et-al:2017,fischetti2018use}.
For a recent survey on mixed-integer optimization techniques in bilevel
optimization see \textcite{Kleinert_et_al:2021}.


\section{Problem Statement}
\label{sec:problem-statement}

We are considering a non-cooperative and complete-information game~$G$
with players indexed by the set~$N = [n]\define \set{1, \dotsc,n}$.
Each player~$i \in N$ solves the optimization problem
\begin{equation}
  \label{opt: player}
  \tag{$\mathcal{P}_i(x_{-i})$}
  \begin{split}
    \min_{x_i} \quad & \pi_i(x_i,x_{-i}) \\
    \st  \quad & x_i \in X_i(x_{-i}) \subseteq  \Z^{k_i} \times \R^{l_i},
  \end{split}
\end{equation}
where $x_i$ is the strategy of player~$i$ and $x_{-i}$ denotes the
vector of strategies of all players except player~$i$.
The function $\pi_i:  \prod_{j \in N}  \R^{k_j+ l_j} \to \R$ denotes
the cost function of player $i$.
The strategy set $X_i(x_{-i})$ of player~$i$ depends on the rivals'
strategies $x_{-i}$ and is a subset of $\Z^{k_i} \times \R^{l_i}$ for
$k_i,l_i \in \nonNegInts$, i.e., the first $k_i$ strategy components
are integral and the remaining $l_i$ are continuous variables.
We assume that the strategy sets are of the form
\begin{align*}
  X_i(x_{-i}) \define\Defset{x_i \in \Z^{k_i} \times
  \R^{l_i}}{g_i(x_i,x_{-i}) \leq 0}
\end{align*}
for a function $g_i: \prod_{j \in N}  \R^{k_j+ l_j} \to \R^{m_i}$
with
components $g_{i,j}$, $j \in [m_i]$, and $m_i \in \nonNegInts$.
We denote by $X(x) \define \prod_{i \in N} X_i(x_{-i})$ the product set
of feasible strategies \wrt $x$ and by
\begin{align*}
  \FeasStrats\define\Defset{x \in \prod_{i \in N} \Z^{k_i} \times
  \R^{l_i}}{x \in X(x)} = \Defset{x \in \prod_{i \in N} \Z^{k_i}
  \times  \R^{l_i}}{g(x) \leq 0}
\end{align*}
the set of feasible strategy profiles, where we abbreviate $g(x)
\define (g_i(x))_{i \in N}$.
We also use its continuous relaxation defined by
\begin{align*}
  \FeasStratsRel\define \Defset{ x \in \prod_{i \in N}\R^{k_i +
  l_i}}{g(x) \leq 0}.
\end{align*}

In order to guarantee that \eqref{opt: player}, our B\&C node problems
\revsecond{and the best-response problems}
 admit an optimal solution (if feasible),\footnote{\revsecond{Here, we assume that we have an appropriate procedure for solving the best-response problem or decide its infeasibility.}}
we make the following standing assumption.
\revsecond{Here, $\FeasStratsRel[-i] \define \defset{x_{-i}}{x \in \FeasStratsRel}$ denotes
the projection of $\FeasStratsRel$ to the rivals' strategies.}
\begin{assumption}\label{ass:GeneralAss} \hfill
  \begin{enumerate}[label=(\roman*)]
  	\item \revsecond{$X(x)$ is compact for every $x \in \FeasStrats$. }\label{ass:GeneralAss:XComp}
  \item $\FeasStrats$ is non-empty and
    $\FeasStratsRel$  is compact. \label[thmpart]{ass:GeneralAss:Poly}
  \item For every player $i$,  her cost function is \revsecond{bounded on $\FeasStratsRel$, lower semicontinuous on
  	$\R^{k_i+l_i} \times \FeasStratsRel[-i]$, and continuous in the rivals'
  	strategies, i.e., $x_{-i} \mapsto \pi_i(y_i,x_{-i})$ is continuous on
  	$\FeasStratsRel[-i]$ for every $y_i \in \R^{k_i+l_i}$.}\label[thmpart]{ass:GeneralAss:Cont}
  \end{enumerate}
\end{assumption}

As usual, a strategy profile~$x^* \in \FeasStrats$ is called a (pure)
Nash equilibrium (NE) if for all players~$i$, the strategy~$x^*_i$
satisfies
\begin{equation*}
  \pi_i(x_i^*,x_{-i}^*) \leq \pi_i(y_i,x^*_{-i})
  \quad \text{ for all } y_i \in X_i(x_{-i}^*),
\end{equation*}
which means that~$x^*_i$ is a strategy minimizing the cost of
player~$i$ parameterized by~$x^*_{-i}$ for all players---such
strategies are called \emph{best responses} (to the
respective~$x_{-i}^*$).

We use the following notation throughout the paper.
We denote the set of all NE by $\Nes
\subset \FeasStrats$ and the tuples~$(x,\pi(x))$ of an NE $x$ and its
corresponding costs, which equal the corresponding best-response values,
$\pi(x)\define (\pi_i(x))_{i\in N}$ by $\Nespi\define \defset{(x,
  \pi(x)) \in \FeasStrats \times \R^N}{x \in\Nes}$ \rev{and refer to those tuples as equilibrium tuples.}

We denote by $x_{i}^\inte \define (x_{i,1}, \dotsc, x_{i,k_i})$ the
integer components of player $i$'s strategy and analogously by
$x_i^\con \define (x_{i,k_i+1}, \dotsc, x_{i,k_i+ l_i})$ the
continuous variables.
By $\FeasStrats[i]\define \defset{x_i}{\exists x_{-i} \text{ with }
  (x_i,x_{-i}) \in \FeasStrats}$ we refer to the  projection of
$\FeasStrats$ to the strategy space of
player $i$.
The projections to the integer, respectively continuous,
components are denoted by $\FeasStratsInt[i] \define
\defset{x_i^\inte}{x_i \in \FeasStrats[i]}$ and $\FeasStratsCon[i]
\define \defset{x_i^\con}{x_i \in \FeasStrats[i]}$.
We use the analogous notation for the entire and partial strategy
profiles~$x$ and~$x_{-i}$, e.g., we abbreviate $x_{-i}^\inte \define
(x_j^\inte)_{j\neq i}$ and $\FeasStratsInt[-i]\define \defset{\xint_{-i}}{ x \in \FeasStrats}$.


\section{The Algorithm}
\label{sec:algorithm}

We now derive a branch-and-cut (B\&C) algorithm to solve the GNEP
defined in \Cref{sec:problem-statement}.
Based on the Nikaido--Isoda (NI) function,
we formulate the search for an NE~$x \in \Nes$
as a mixed-integer bilevel problem.
The NI function is given by
\begin{equation*}
  \Psi(x,y) =
  \sum_{i \in N} \pi_i(x) - \sum_{i \in N} \pi_i(y_i,x_{-i})
\end{equation*}
and we further define
\begin{equation*}
  \hat{V}(x) = \max_{y \in X(x)} \ \Psi(x,y).
\end{equation*}
Here, the $y$-variables are maximized so that $y_i$ is a best response
to $x_{-i}$.
Thus, the $\hat{V}$-function computes the aggregated
regrets of players w.r.t.\ the current strategy profile $x$.
It has the well-known property that for all $x \in \FeasStrats$, the inequality
$\hat{V}(x)\geq 0$ holds, and, that $\hat{V}(x) = 0$ is equivalent to
$x$ being an NE, because, in this case, all strategies~$x_i$ are
best responses; see, e.g., \textcite{facchinei2010generalized}.

Consequently, we are looking for a global minimizer~$x$ of the
$\hat{V}$-function, i.e., we want to solve the problem
\begin{equation*}
  \min_{x \in \FeasStrats} \ \hat{V}(x)
  = \min_{x \in \FeasStrats} \Set{ \max_{y \in X(x)} \  \sum_{i \in N} \pi_i(x)
    - \sum_{i \in N} \pi_i(y_i,x_{-i}) }.
\end{equation*}

By recognizing that we are interested in the sum of best-response values
and not in the $y$-variables themselves, we formulate this minimization
problem as a bilevel problem with no $y$-variables appearing
explicitly.
This leads to the epigraph reformulation
\begin{equation}
  \tag{$\mathcal{R}$}
  \label{model:R}
  \begin{split}
    \min_{x \in \FeasStrats,\proxy \in \R^N} \quad
    & \sum_{i \in N} \pi_i(x) - \proxy_i \\
    \st  \quad
    & \proxy \leq \VBR(x),
  \end{split}
\end{equation}
where  $\VBR(x)$ denotes the vector of best-response values, i.e.,
$\VBR(x) \define (\Phi_i(x_{-i}))_{i \in N}$ with
\begin{equation*}
  \Phi_i(x_{-i})
  \define \min_{y_i \in X_i(x_{-i})} \ \pi_i(y_i, x_{-i}).
\end{equation*}
Seen as a bilevel problem, $\sum_{i \in N}\Phi_i(x_{-i})$ is the
optimal-value function of the lower-level problem (in~$y$).

The corresponding high-point relaxation (HPR) of model \eqref{model:R}
is then given by
\begin{equation}
  \tag{HPR}
  \label{model:HPR}
  \begin{split}
    \min_{x \in \FeasStrats, \proxy \in \R^N} \quad
    & \sum_{i \in N} \pi_i(x) - \proxy_i \\
    \st \quad
    &\proxy\leq \proxy^+,
  \end{split}
\end{equation}
where we use $\proxy^+ \in \R^N$ as a valid and finite upper-bound
vector for $\proxy$ to ensure boundedness of the problem.
An example of a valid upper bound is given by
\begin{equation*}
  \proxy_i^+ = \sup_{x \in \FeasStrats}  \pi_i(x),
\end{equation*}
where we note that this value is finite by the
assumption that $\pi_i$ is bounded.

Using the continuous relaxation $\FeasStratsRel$ of $\FeasStrats$ in
\eqref{model:HPR} then leads to the continuous high-point relaxation
\begin{equation}
  \tag{C-HPR}
  \label{model:CHPR}
  \begin{split}
    \min_{x \in \FeasStratsRel,\proxy \in \R^N} \quad
    & \sum_{i \in N} \pi_i(x) - \proxy_i \\
    \st \quad & \proxy \leq \proxy^+.
  \end{split}
\end{equation}

In a B\&C tree, each node problem contains the root node problem,
\eqref{model:CHPR} in our case, with additional constraints.
These constraints correspond to the branching constraints and cuts
added along the path from the root node to node~$t$, denoted by $B_t$
and $C_t$, respectively.
Hence, the problem at node~$t$ can be formulated as
\begin{equation}
  \tag{$\mathcal{R}_t$}
  \label{model:node_t}
  \begin{split}
    \min_{x,\proxy} \quad
    & \sum_{i \in N} \pi_i(x) - \proxy_i
    \\
    \st \quad  &(x,\proxy) \in \FeasRt\define  \left(\FeasStratsRel \times
      \Set{\bar{\proxy}\in\R^N\mid \bar{\proxy}\leq \proxy^+} \right)
    \cap B_t \cap C_t
  \end{split}
\end{equation}

We now have to discuss the required cuts in more detail and start with
the following definition.

\begin{definition}
  \label{def:non-ne-cut}
  For any node~$t$ of the search tree, let $(x^*, \proxy^*)$ be an
  integer-feasible node solution, i.e., an integer-feasible solution
  to \eqref{model:node_t}, with a non-positive objective value and
  let $x^*$ not be an NE.
  Consider further an arbitrary corresponding best response

  \begin{equation*}
    y^* \in \argmin_{y \in X(x^*)} \ \sum_{i \in N}\pi_i(y_i,x_{-i}^*).
  \end{equation*}
  Then, we call an inequality $c(x,\proxy;x^*,\proxy^*,y^*)\leq 0$,
  which is parameterized by $(x^*, \proxy^*, y^*)$, a
  \emph{non-NE-cut} (for node $t$) if the following \revsecond{three} properties are
  satisfied:
  \begin{enumerate}[label=(\roman*)]
  \item\label{def:non-ne-cut:1} It is satisfied by all points
    $(x, \pi(x)) \in \Nespi \cap B_t \cap C_t$.
  \item\label{def:non-ne-cut:2} It is violated by $(x^*, \proxy^*)$.
  \item\label{def:non-ne-cut:3} \revsecond{The set $\defset{(x,\proxy) \in
        \FeasStratsRel\times\R^N}{ c(x,\proxy;x^*,\proxy^*,y^*)\leq
        0}$ is closed.}
  \end{enumerate}
\end{definition}

In addition, a non-NE-cut is said to be globally valid if it is
satisfied by all points $(x, \pi(x)) \in \Nespi$.
Such a cut is then valid for any node~$t$ of the B\&C search tree.

The B\&C method now works as follows.
Starting from the root node, we solve the current node problem
\eqref{model:node_t}. In case that the problem is infeasible or admits
an objective value strictly larger than zero, there does not exist an
equilibrium in this node and we prune it. Otherwise, we check the
optimal node solution for integer-feasibility and create new nodes as
usual by branching on fractional integer variables if necessary. Once
we obtain an integer-feasible node solution, we check if it actually
is an NE. If so, we stop and return the NE. Otherwise, we  use
a non-NE-cut  to cut the integer-feasible point
without removing any potential NE contained in the node.
The procedure to process a node $t$ is described formally in \Cref{algorithm:BC}.

\begin{algorithm}
  \caption{Processing Node~$t$}
  \label{algorithm:BC}
  \begin{algorithmic}[1]
    \State Solve \eqref{model:node_t}. \label{line:step_1}
    \If{\eqref{model:node_t} is infeasible \textbf{or} the optimal objective is strictly positive} \label{line: Infeasible}
    \State Prune the node.
    \Else
    \State Let $(x^*,\eta^*)$ be a solution of \eqref{model:node_t}. \label{line: NodeSol}
    \If{$x^* \notin \FeasStrats$} \label{line: xNotInX}
    \State Create two child nodes by branching on a fractional variable.
    \Else
    \State Solve the best-response problems for $x^*$ to obtain a solution~$y^*$. \label{line: SolveBR}
    \If{$\Psi(x^*,y^*) = 0$} \label{line: NEFound}
    \Comment{$x^*$ is an NE}
    \State Return $x^*$ and stop the overall B\&C method.\label{line: return-x-star}
    \Else \Comment{$\eta^* \nleq \Phi(x^*)$ \& $x^*$ is not an NE}  \label{line:
    	proxy>BR}
    \State Augment $C_t$ with a non-NE-cut
    $c(x,\proxy;x^*,\eta^*,y^*)\leq 0$.\label{line: Cut}
    \State Go to Step~\ref{line:step_1}.
    \EndIf
    \EndIf
    \EndIf
  \end{algorithmic}
\end{algorithm}%
%
%
%
\revsecond{Note that the solution in Line~\ref{line: NodeSol}
  always exists: $\FeasStratsRel$ is compact, $B_t$ is closed,
  \revsecond{$(\FeasStratsRel\times\R^N)\cap C_t$ is closed by
    Definition~\ref{def:non-ne-cut}\ref{def:non-ne-cut:3},}  the
  objective is lower semicontinuous, and, since $\pi$ is bounded and
  $\proxy \leq \proxy^+$, the objective tends to $+\infty$ if
  $\proxy_i \to -\infty$ for some $i \in N$, so that all sublevel sets are
  compact.}
See Appendix~\ref{appendix:BC-example} for an exemplary application of
the resulting B\&C method when Algorithm~\ref{algorithm:BC} is used to
process the nodes.

\begin{remark}
  \label{rem: AggCuts1}
  It is also possible to only introduce a single epigraph
  variable~$\proxyagg$ instead of a vector with one epigraph variable
  for each player~$i \in N$.
  This would result in the alternative formulation
  \begin{equation}
    \tag{$\mathcal{R}'$}
    \label{model:RPlayer}
    \begin{split}
      \min_{x \in \FeasStrats, \proxyagg \in \R} \quad
      & \left( \sum_{i \in N} \pi_i(x) \right) - \proxyagg
      \\
      \st \quad
      & \proxyagg \leq \sum_{i \in N}\VBR_i(x_{-i}) .
    \end{split}
  \end{equation}
  \Cref{algorithm:BC} can be adapted to this setting by changing
  \ref{def:non-ne-cut:1} in \Cref{def:non-ne-cut} to
  \begin{enumerate}[label=(\roman*')]
  \item\label{rem: AggCuts1:non-NE-cutCon1}
    It is satisfied by all tuples $(x,\sum_{i \in N}\pi_i(x)) \in
    (\Nes \times \R) \cap C_t\cap B_t$.
  \end{enumerate}
  We will later see that the disaggregated version has favorable
  properties (see, e.g., \Cref{rem: StandNECuts}), which is why we focus
  on the latter in what follows.
\end{remark}

\revised{
  We now come to the correctness of the B\&C method.

  \begin{theorem}\label{thm:correctness_algBC}
    The B\&C algorithm is correct, i.e., if it terminates,
    it yields an equilibrium $x^*\in \Nes$ or a certificate for the
    non-existence of NE.
  \end{theorem}
  \begin{proof}
    Assume that the B\&C algorithm terminates, i.e.,
    either
    Algorithm~\ref{algorithm:BC} returns a point $x^*$ in
    Line~\ref{line: return-x-star} or every node was ultimately pruned.
    \begin{description}[leftmargin=10pt]
    \item[Returned $x^*$]
      We argue that $x^*$ is an NE.
      If \Cref{algorithm:BC} reaches Line~\ref{line: return-x-star},
      it holds that $x^* \in \FeasStrats$ and
      \begin{equation*}
        \hat{V}(x^*)
        = \sum_{i \in N} \pi_i(x^*)
        - \min_{y \in X(x^*)} \ \sum_{i \in N} \pi_i(y_i,x^*_{-i})
        = \Psi(x^*,y^*) = 0.
      \end{equation*}
      Here, the first equality is valid due to the definition of $\hat{V}$
      and the second one holds due to the optimality of $y^*$. Hence, we have
      $\hat{V}(x^*) = 0$ and $x^* \in \FeasStrats$, which is equivalent to
      $x^*$ being an NE.

    \item[Every node pruned]
      We start with two observations:
      First, if the optimal objective value of \eqref{model:node_t} is
      strictly positive, then $\Nespi \cap \FeasRt = \emptyset$
      holds. This is because the objective value of $(x, \pi(x))$ for
      \eqref{model:node_t} and any NE~$x \in \Nes$ is equal to zero.

      The second observation concerns an invariant of the algorithm:
      Due to Condition~\ref{def:non-ne-cut:1} in  \Cref{def:non-ne-cut}, the
      following invariant is true throughout the execution of the B\&C
      algorithm.
      The set of equilibria with corresponding best-response values
      $\Nespi= \Defset{(x,\pi(x)) \in (\FeasStrats \times \R^N)}{x \in
        \mathcal{E}}$ is contained in the union of the feasible
      sets of the problems~\eqref{model:node_t} over all leaf nodes $t$ in
      the B\&C tree, i.e.,
      \begin{align*}
        \Nespi
        \subseteq  \bigcup_{t \text{ is a leaf}} \FeasRt.
      \end{align*}
      Note that pruned nodes are leaves of the B\&C tree as well.

      With these observations, we argue now that in this case, no NE can exist if
      the B\&C method terminates due to every node being ultimately pruned.
      To do so, we show in the following that $\FeasRt\cap \Nespi = \emptyset$
      holds for all leaves $t$ in the case of the B\&C algorithm
      terminating without finding an NE.
      It then follows directly by the above invariant that there does not
      exist any NE.
      If
      every node~$t$ was ultimately pruned, then the condition in
      Line~\ref{line: Infeasible} was met and Problem~\eqref{model:node_t}
      became either infeasible or had a strictly positive optimal
      objective value.
      In the former case, it is clear that $\FeasRt\cap \Nespi =
      \emptyset$ holds because of $\FeasRt = \emptyset$.
      In the latter case, the claim is a direct consequence
      of our first observation.
      Thus, the proof is finished.\qedhere
    \end{description}
  \end{proof}
}

A variant of Algorithm \ref{algorithm:BC}, where cuts are derived
before having an integer-feasible solution satisfies the same
correctness result. However, preliminary numerical experiments showed
that this variant of the overall method is less efficient.

So far, we have shown the correctness of our B\&C method for arbitrary
non-NE-cuts. Clearly,
whether or not the B\&C method terminates after
a finite number of steps depends on the utilized non-NE-cuts.
In this regard, we propose in the following section several cuts and
give sufficient conditions under which they lead to a finite
termination of the B\&C method.


\section{Cuts and Finite Termination}
\label{sec:cuts-finite-termination}

\rev{In this section, we investigate the existence of non-NE-cuts
  as well as provide sufficient conditions for finite termination of
  our B\&C method.
  We require the following useful observation
  regarding \Cref{algorithm:BC}.
  Whenever a cut needs to be added, at least one entry of the proxy
  variable $\proxy$ must exhibit slack with respect to the
  best-response value.
  This insight is crucial to show the existence of a proper
  cut that tightens the respective proxy variable.}
  \revsecond{From now on, let us denote by $N(x^*,\proxy^*)\define
  	\Defset{i \in N}{\proxy^*_i > \VBR_i(x_{-i}^*)}$ the set of players whose corresponding
  	proxy variable admits slack \wrt a node optimal solution $(x^*,\proxy^*)$.}

\revised{
  \begin{lemma}\label{lem: proxy>BR}
    If the algorithm enters the else-part in Line~\ref{line:
      proxy>BR}, then $N(x^*,\proxy^*) \neq \emptyset$.
  \end{lemma}
  \begin{proof}
    Since the condition in Line~\ref{line: Infeasible} was not met,
    the objective value of $(x^*,\proxy^*)$ is non-positive.
    Moreover, by the condition in Line~\ref{line: NEFound}   not being satisfied,
    we know that $\NI(x^*,y^*) \neq 0$. Since $y^*$ is a best response to $x^*$,
    it follows that $\NI(x^*,y^*) > 0$. With these two observations, we now get
    \begin{align*}
      \sum_{i\in N} \pi_i(x^*)-\VBR_i(x^*_{-i}) = \NI(x^*,y^*) > 0
      \geq \sum_{i\in N} \pi_i(x^*)-\proxy_i^* .
    \end{align*}
    Subtracting $\sum_{i\in N} \pi_i(x^*)$ from both sides and
    multiplying with $-1$ yields
    \begin{align*}
      \sum_{i\in N} \proxy_i^*> \sum_{i\in N}  \VBR_i(x^*_{-i}),
    \end{align*}
    from which the statement follows immediately.
  \end{proof}
}

\rev{The remaining section is split into \revsecond{three} subsections.
  First, we consider the
  case of standard NEPs for which we derive non-NE-cuts via best-response
  cuts.}
\revsecond{Remarkably, these cuts always exist and do not require
  any further assumptions.}
\rev{For these cuts, we then provide sufficient conditions for the finite termination
  of our B\&C method.
  Next, we tackle the GNEP setting and provide suitable
  assumptions under which we can guarantee the existence of non-NE-cuts via
  intersection cuts.}
\revsecond{Finally, under additional assumptions, we present a fallback strategy
  that is invoked whenever the ICs become too shallow and thereby guarantees finite termination. }

\subsection{Standard Nash Equilibrium Problems}
\label{sec:cuts-ne}

In this section, we consider the special case of $G$ being a standard
NEP, i.e., $X_i(x_{-i}) \equiv X_i$ for some fixed strategy set $X_i$
given by $X_i \define \defset{x_i \in \Z^{k_i}\times \R^{l_i}}{g_i(x_i)
  \leq 0}$.
Note that the set of feasible strategy profiles is then
given by $\FeasStrats= \prod_{i \in N}X_i$.

\begin{lemma}\label{lem: Cut}
  In the situation of \Cref{def:non-ne-cut}, the best-response cut given by
  \begin{align}
    \label{eq: PlayerCut}
    c(x,\proxy;
    x^*, \proxy^*, y^*) \define c_i(x,\proxy;y^*) \define \proxy_i -
    \pi_i(y_i^*,{x}_{-i}) \leq 0
  \end{align}
  yields a non-NE-cut for every  $i \in N(x^*,\proxy^*)$.
  \revised{Note that $N(x^*,\proxy^*) \neq \emptyset$ holds in the situation of
    Line~\ref{line: proxy>BR} by \Cref{lem: proxy>BR}, i.e., we can
    always use a non-NE-cut of the form \eqref{eq: PlayerCut}.}
  Moreover, these cuts are globally valid, i.e., they are
  satisfied for all $(x,\pi(x)) \in \Nespi$.
\end{lemma}
\begin{proof}
  For an arbitrary tuple $(\bar{x},\pi(\bar{x})) \in \Nespi$,
  the equilibrium condition implies $\pi(\bar{x}) =
  \VBR(\bar{x})$ and, hence, for any $i \in N$ we have
  \begin{equation*}
    \pi_i(\bar{x})
    = \min_{y_i \in X_i} \pi_i(y_i,\bar{x}_{-i}) \leq
    \pi_i(y_i^*,\bar{x}_{-i}).
  \end{equation*}
  Thus, $c_i(\bar{x},\pi(\bar{x});y^*) \leq 0$
  holds for all $i \in N\supseteq N(x^*,\proxy^*)$.
  In particular, the cuts in~\eqref{eq: PlayerCut} fulfill
  Condition~\ref{def:non-ne-cut:1} of \Cref{def:non-ne-cut}.

  For Condition~\ref{def:non-ne-cut:2}, we get as an immediate
  consequence of $i \in N(x^*,\proxy^*)$ the inequality
  \begin{equation*}
    \proxy_i^* >
    \VBR_i(x^*_{-i}) = \min_{y_i \in X_i}\pi_i(y_i,x_{-i}^*) =
    \pi_i(y_i^*,x_{-i}^*).
  \end{equation*}

  \revsecond{Finally, Condition~\ref{def:non-ne-cut:3} is satisfied since
    $c_i(\cdot,\cdot;y^*)$ is continuous on $\FeasStratsRel\times\R^N$
    by Assumption~\ref{ass:GeneralAss}\ref{ass:GeneralAss:Cont}.}
  \qedhere
\end{proof}

\begin{remark}\label{rem: StandNECuts}
  \begin{thmparts}
  \item   The first part of the proof above does not only work for
    \mbox{$(\bar{x},\pi(\bar{x})) \in \Nespi$} but for all $(x,\VBR(x))
    \in \FeasStrats\times \R^N$.
    Hence, if one only uses the cuts in \Cref{lem: Cut}, then
    \begin{align*}
      \Defset{(x,\VBR(x)) \in \FeasStrats\times \R^N}{\exists \proxy \in
      \R^N \text{ with } (x,\proxy) \in F_t } \subseteq F_t
    \end{align*}
    holds.
    In particular, an integer-feasible node solution~$(x^*,\proxy^*)$
    with $\proxy^*\leq \VBR(x^*)$ always satisfies $\proxy^* =
    \VBR(x^*)$.

  \item  For the situation considered in \Cref{rem: AggCuts1}, it follows
    analogously to \Cref{lem: Cut} that a non-NE-cut is induced by an
    aggregated version of the functions defined in \Cref{lem: Cut},
    i.e., by the function
    \begin{equation*}
      c(x,\proxyagg;y^*) \define \proxyagg
      - \sum_{i \in N}\pi_i(y_i^*,{x}_{-i}).
    \end{equation*}
    Note that these aggregated cuts are weaker than the individual
    ones in the following sense:
    the aggregated  cut
    \begin{align}
    	\label{eq: AggCutOurSetting}
    	c(x,\proxy;y^*) \define
    	\sum_{i \in N}\proxy_i - \sum_{i \in N}\pi_i(y_i^*,{x}_{-i})
    \end{align}
     is implied by the $n$
    inequalities~\eqref{eq: PlayerCut}, whereas a point $(x,\proxy)$ with
    $\proxy_i > \pi_i(y_i^*,x_{-i})$ for some $i \in N$ and
    $\sum_{i \in N} \proxy_i \leq \sum_{i \in N} \pi_i(y_i^*,x_{-i})$
    violates \eqref{eq: PlayerCut} for player~$i$ but
    satisfies~\eqref{eq: AggCutOurSetting}.

  \item \textcite{dragotto2023zero} introduce cuts similar to the ones
    proposed in \Cref{lem: Cut}.
    They consider so-called integer-programming games (IPGs), i.e.,
    standard NEPs as discussed in this section with the additional
    properties of all strategies being integer and $g_i$ being affine.
    In order to solve such an IPG, the authors derive a cutting-plane
    algorithm in which the space of strategy profiles is reduced via
    cuts of the form
    \begin{equation*}
      c_i(x;y^*) \define \pi_i(x) - \pi_i(y_i^*,x_{-i})\leq 0,
      \quad x \in \FeasStrats,
    \end{equation*}
    for a best response $y^*$ w.r.t.\ $x$ and $i \in N$ with $ \pi_i(x) >
    \pi_i(y_i^*,x_{-i})$.
    Note however, that in contrast to our approach, they (i) solve
    in every iteration a mixed-integer optimization problem,
    (ii) only consider
    standard Nash games, and (iii) are restricted to the pure integer
    setting.

  \item  A no-good cut is an inequality that excludes exactly one integer point from the feasible set.
    In the binary case, a no-good cut cutting off $x^*$ is
    simply given by
    \begin{equation*}
      \sum_{j:x^*_j=0} x_j + \sum_{j:x^*_j=1} (1-x_j) \geq 1.
    \end{equation*}
    It can be extended to the bounded integer case by using a binary
    expansion of the integer variables.
    In the special case of $X_i \subseteq \mathbb{Z}^{k_i}$ for all $i$,
    simple no-good cuts are trivially globally-valid non-NE-cuts.
  \end{thmparts}
\end{remark}

In the following, we derive sufficient conditions under which
\Cref{algorithm:BC}
terminates in finite time when using the cuts introduced in \Cref{lem:
  Cut}.
\rev{Note that this is sufficient for the finite termination of  the overall B\&C method
  as the set of feasible strategy profiles is bounded by our standing
  Assumption~\ref{ass:GeneralAss} and thus only finitely many nodes
  during the execution of our B\&C may appear.}
The following theorem provides an abstract sufficient condition for the finite
termination of \Cref{algorithm:BC}.
In the subsequent lemmas, we show that this condition is satisfied
for the two important special cases in which
\begin{enumerate}[label=(\roman*)]
\item the players' cost functions are concave in their own continuous
  strategies \revsecond{and all strategy sets are polyhedral} or
\item the players' cost functions only depend on their own strategy and
  the rivals' integer strategy components.
\end{enumerate}

In order to state the promised theorem, we introduce the following
terminology.
Let us denote by
\begin{align*}
  \BR(x) \define \argmin \Defset{\sum_{i \in
  N}\pi_i(y_i,x_{-i})}{y \in \prod_{i \in N} X_i}
\end{align*}
the set of best responses
to $x \in \FeasStrats$.
Moreover, let us define the set of all possible best-response sets
by $\BRcom \define \defset{\BR(x)}{x \in \FeasStrats} \subseteq
\powset(\FeasStrats)$, where we denote by $\powset(\FeasStrats)$ the
power set of $\FeasStrats = \prod_{i \in N} X_i$.

\begin{theorem}
  \label{thm: FinitelyManyCutsGeneral}
  Assume that $\abs{\BRcom}$ is finite.
  If we use the non-NE-cut \eqref{eq: PlayerCut} from \Cref{lem: Cut}
  in Line~\ref{line: Cut} of \Cref{algorithm:BC}, then \Cref{algorithm:BC}
  terminates after a finite number of steps.
\end{theorem}
\begin{proof}
  Consider an arbitrary sequence of iterations of
  \Cref{algorithm:BC} with corresponding optimal solutions
  $(x^*_s,\proxy^*_s), s= 1,\ldots,j$, and computed best
  responses~$y_s^*$, $s= 1,\ldots,j$, in Line~\ref{line: SolveBR}.
  Moreover, let $i_s$ be the index of the player for which $C_t$ was
  augmented with a cut from \Cref{lem: Cut}  in the $s$-th iteration
  for every $ s= 1,\ldots,j$.
  The following claim holds.
  \begin{claim}\label{claim: iOptimalAfterCut}
    For any two iteration indices $1 \leq s_1 < s_2\leq j$ with $\BR(x^*_{s_1})
    = \BR(x^*_{s_2})$,  we have $i_{s_1} \notin
    N(x^*_{s_2},\proxy^*_{s_2})$.
  \end{claim}
  \begin{proofClaim}
    Since $(x^*_{s_2},\proxy^*_{s_2})$ is feasible for
    \eqref{model:node_t}, we  have in particular
    $(x^*_{s_2},\proxy^*_{s_2}) \in C_t$ and, hence,
    \begin{align*}
      0 \geq c_{i_{s_1}}(x^*_{s_2},\proxy^*_{s_2};y^*_{s_1})
      & = (\proxy^*_{s_2})_{i_{s_1}} -
        \pi_{i_{s_1}}((y^*_{s_1})_{i_{s_1}},(x^*_{s_2})_{-i_{{s_1}}})
      \\
      & = (\proxy^*_{s_2})_{i_{s_1}} -
        \VBR_{i_{s_1}}((x^*_{s_2})_{-i_{s_1}}),
    \end{align*}
    where the last equality is valid by $y^*_{s_1} \in \BR(x^*_{s_1})
    = \BR(x^*_{s_2})$.
    Thus, $i_{s_1} \notin N(x^*_{s_2},\proxy^*_{s_2})$, which shows
    the claim.
  \end{proofClaim}
  We get as a consequence the following statement.
  \begin{claim}
    For any $\mathcal{B}\in \BRcom$, there may exist at most $n$
    sequence indices $s_1 < \dotsb < s_n < j$ with  $\BR(x^*_{s_l}) =
    \mathcal{B}, l = 1,\ldots,n$.
  \end{claim}
  \begin{proofClaim}
    \revised{
      Suppose for the sake of a contradiction that
      there is a sequence of $n+1$ iterations $s_1 < \dotsb <s_{n+1}$
      with $\BR(x^*_{s_l}) =
      \mathcal{B}, l = 1,\ldots,n+1$. Since we have $n$ players,
      the pigeon-hole principle implies that  we added a best-response cut
      for a player twice, i.e., there has to exist $i\in N$ and $l_1,l_2 \leq n+1$ with
      $i_{s_{l_k}} = i$ and $i \in N(x^*_{s_{l_k}},\proxy^*_{s_{l_k}})$ for both $k = 1,2$.
      This contradicts \Cref{claim: iOptimalAfterCut}.
    }
  \end{proofClaim}

  From the above claim, it now follows directly that $j$ cannot be
  arbitrarily large as $\BRcom$ contains only finitely many
  best-response sets $\mathcal{B}$.
\end{proof}

In the following two lemmas, we present two applications for the
above \namecref{thm: FinitelyManyCutsGeneral}.
In \Cref{lem: FinitelyManyBRSetsConcave}, we show that  \Cref{thm:
  FinitelyManyCutsGeneral} is applicable for games in which the
players' cost function are concave in their own continuous strategies
and in which all strategy sets are polyhedral.

\begin{lemma}
  \label{lem: FinitelyManyBRSetsConcave}
  Assume that the following holds for all $i \in N$:
  \begin{enumerate}[label=(\roman*)]
  \item Player $i$'s strategy set is given by $X_i = \defset{x_i\in \Z^{k_i}\times
      \R^{l_i}}{A_ix_i\leq b_i}$ for some matrix~$A_i$ and some
    vector~$b_i$.
  \item Player $i$'s cost function is concave in her continuous variables
    $x_i^\con$, i.e., for all \mbox{$x_i^\inte \in \FeasStratsInt[i]$} and
    $x_{-i} \in \FeasStrats[-i]$, the function
    \begin{equation*}
      \pi_i(x_i^\inte,\cdot,x_{-i}):\set{x_i^\con : (x_i^\inte, x_i^\con) \in X_i} \to \R,
      \quad
      x_i^\con \mapsto \pi_i(x_i^\inte,x_i^\con,x_{-i}),
    \end{equation*}
    is concave.
  \end{enumerate}
  Then, $\abs{\BRcom} < \infty$ holds and the overall B\&C method
  terminates in finite time if the cuts of \Cref{lem: Cut} are used.
\end{lemma}
\begin{proof}
  Note that for any $x \in \FeasStrats$, we can write
  \begin{equation*}
    \BR(x)
    = \bigcup_{\hat{y}^\inte \in \FeasStratsInt}
    \Set{\hat{y}^\inte} \times \BR(x,\hat{y}^\inte)
  \end{equation*}
  with
  \begin{equation*}
    \BR(x,\hat{y}^\inte) \define
    \Defset{y^\con}{ (\hat{y}^\inte,y^\con) \in \BR(x)}.
  \end{equation*}
  The set $\BR(x,\hat{y}^\inte)$ is either empty if there are no
  $(\hat{y}^\inte,y^\con) \in \BR(x)$ or it corresponds to the set of
  optimal solutions to the optimization problem
  \begin{equation*}
    \min_{y^\con} \quad
    \sum_{i \in N} \pi_i(\hat{y}_i^\inte,y_i^\con,x_{-i})
    \quad \st \quad y^\con \in X[\hat{y}^\inte]
  \end{equation*}
  with
  \begin{equation*}
    X[\hat{y}^\inte] \define
    \Defset{y^\con}{A_i^\con y_i^\con \leq b_i - A_i^\inte
      \hat{y}_i^\inte, \, i \in N}.
  \end{equation*}
  Here, we use the notations $A_i^\con$ and $A_i^\inte$ to denote the
  sub-matrices of $A_i$ that correspond to the continuous and integral
  strategy components, i.e., $A_iy_i = A_i^\con y_i^\con + A_i^\inte
  y_i^\inte$.
  By the concavity assumption on $\pi_i$, we can exploit the fact that
  the set of optimal solutions of the above optimization problem is
  the union of some faces of $X[\hat{y}^\inte]$; see \Cref{lem:
    HelpConc}.
  In particular, $\BR(x,\hat{y}^\inte)$ is (for any $x \in
  \FeasStrats$) the union of some faces of $X[\hat{y}^\inte]$.
  Note that this is in particular true if the set is empty.
  Hence, we get
  \begin{align*}
    \BRcom \subseteq
    \Defset{ \bigcup_{\hat{y}^\inte \in \FeasStratsInt}
    \{\hat{y}^\inte\} \times \Fa(\hat{y}^\inte)}{\Fa(\hat{y}^\inte)
    \text{ is the union of faces of } X[\hat{y}^\inte]}.
  \end{align*}
  The latter set is finite as $\FeasStratsInt$ is finite (by
  $\FeasStratsRel$ being bounded) and
  any polyhedron only has finitely many faces and thus also finitely
  many different unions of them.
\end{proof}

Next, we show in the following \namecref{lem:
  FinitelyManyBRSetsInteger} that \Cref{thm: FinitelyManyCutsGeneral}
is also applicable for games in which  the cost functions of players
only depend on their own strategy and the rivals' integer strategy
components.

\begin{lemma}
  \label{lem: FinitelyManyBRSetsInteger}
  Assume that for all $i \in N$,
  player $i$'s cost function  $\pi_i$ only depends on $x_i$ and the rivals'
  integer strategy components, i.e., there exists a function
  $\pi_i^\inte: X_i\times \FeasStratsInt[-i] \to \R$ such that $\pi_i(x) =
  \pi_i^\inte(x_i,x_{-i}^\inte)$ for all $x \in \FeasStrats$.
  Then, $\abs{\BRcom} < \infty$ holds and the overall B\&C method
  terminates in finite time if the cuts of \Cref{lem: Cut} are used.
\end{lemma}
\begin{proof}
  By the assumptions of the structure of the cost functions, it is
  clear that $\BR(x)$ only depends on the integer components of $x$,
  i.e., $\BR(x) = \BR(\tilde{x})$ for all $x,\tilde{x}\in \FeasStrats$ with
  $x^\inte = \tilde{x}^\inte$. Hence, we can define $\BR(\xint)\define \BR(\tilde{x})$ for all $\xint \in \FeasStratsInt$ and
  $\tilde{x}\in \FeasStrats$ with $\tilde{x}^\inte = \xint$.
  Then, we have
  \revised{
    \begin{equation*}
      \BRcom=\Defset{\BR(x)}{x \in \FeasStrats} = \Defset{\BR(\xint)}{\xint \in
        \FeasStratsInt}.
    \end{equation*}}%
  Hence, it follows that $\BRcom$ is finite as $\FeasStratsInt$ is finite.
  The latter is implied by our assumption that $\FeasStratsRel$ is bounded.
\end{proof}

\rev{We close this section with a remark on the different
  assumptions of the last two lemmas.}
\begin{remark}\label{rem: Ride}
  \rev{The assumptions of the first lemma apply to the broad class of
    mixed-integer games with concave cost functions. While these
    functions obviously include linear functions, the case of strictly
concave functions also appears. For example, when economies of scale
effects are present, i.e., situations in which marginal costs are
decreasing. Examples include network design games, where cost
functions are modeled by fixed or concave costs; see, e.g.,
\textcite{Anshelevich08,FalkeHarks13}.}
  \rev{The second lemma can be applied to any pure-integer
    game with nonlinear (not necessarily integral polyhedral) strategy
    sets.
    Such situations arise naturally in, e.g., production games for
    indivisible goods in which firms aim to determine the optimal
    integer amount of produced goods subject to nonlinear production
    constraints that arise, for example, due to nonlinear resource
    consumption profiles subject to budget constraints.}
\end{remark}


\subsection{Generalized Nash Equilibrium Problems}
\label{sec:cuts-gne}

In~\textcite{fischetti2018use}, the authors transfer the idea of
intersection cuts (ICs) to bilevel optimization, which
were originally introduced by \textcite{balas1971intersection} in the
context of integer programming.
We follow this approach and derive non-NE-cuts
(cf.~\Cref{def:non-ne-cut}) for GNEPs via ICs.
For the remainder of this section, consider the situation of
\Cref{def:non-ne-cut} (resp.~Line~\ref{line: Cut} of
\Cref{algorithm:BC}) and fix the corresponding integer feasible
solution $(x^*,\proxy^*)$ and a corresponding best response~$y^*$.
With this at hand, we derive sufficient conditions to define a
non-NE-cut via an IC. For this, we need to guarantee the existence of
the following two objects:
\begin{enumerate}[label=(\roman*)]
\item a cone $\Corner$ pointed at $(x^*,\proxy^*)$ containing
  $\Nespi\cap C_t\cap B_t$
\item and an \emph{NE-free set} $\freeset(x^*,\proxy^*)$, i.e., a
  convex set whose interior contains the point $(x^*,\proxy^*)$
  but no point of $\Nespi\cap C_t\cap B_t$.
\end{enumerate}

\rev{Given such a cone and NE-free set, an IC exists that is a
  non-NE-cut, i.e., it is valid for all equilibrium tuples in
  $\Nespi\cap C_t\cap B_t$
  but cuts off $(x^*,\proxy^*)$, \revsecond{and it defines a
  closed feasible set}. We refer to standard textbooks such as
  \textcite[Chapter~6]{Conforti-et-al:2014}.
  Note that the construction by \textcite[Chapter~6]{Conforti-et-al:2014}
  is carried out in the MILP setting but extends
  to more general situations. This is due to the intersection cut
  only depending on the structure of the required convex set and cone, but
  not on the set of points which must not be
  cut off; see \textcite{Conforti_CutGenerating}.
  In order to keep the exposition self-contained,
  we describe in Appendix~\ref{app: ICConstruction} the explicit
  construction of an IC. There, we also show in Lemma~\ref{lem:
    ICCutStandNEP} that the IC applied to the NEPs is at most as
  strong as the best-response cut defined in the previous section
  and coincides with the latter under suitable assumptions.}

We start with a discussion of the existence of a suitable cone
$\Corner$.
For simplicity, we consider the case in which $g_i(x) = A_i x - b_i$
holds for a suitable matrix~$A_i$ and a vector~$b_i$.
Hence, we consider a polyhedral setting.
Since ICs are linear cuts, the set of feasible solutions $\FeasRt$ in
a node problem remains a polyhedron if we only employ ICs as
non-NE-cuts. In this regard,
in case that $(x^*,\proxy^*)$ is a vertex of $\FeasRt$, we can simply
use the corresponding LP basis cone for $\Corner$.
For general cost functions, this is of course not guaranteed but if
$x \mapsto \sum_{i \in N} \pi_i(x)$ is concave, then
\eqref{model:node_t} admits an optimal solution at a vertex which can
be chosen as $(x^*,\proxy^*)$.

\begin{remark}
  Let us also remark that for the general case in which
  $(x^*,\proxy^*)$ is not a vertex of $\FeasRt$, we can branch
  sufficiently often until $(x^*,\proxy^*)$ becomes a vertex. To see
  this, just observe that $(x^*,\proxy^*)$ is a vertex of the set
  $\FeasRt \cap \defset{ (x,\proxy)}{(x,\proxy)_j \sim_j
    (x^*,\proxy^*)_j \text{ for all } j}$ for any inequality given by
  $\sim_j\, \in \set{\leq,\geq}$.
  In this regard, a slight variant of \Cref{algorithm:BC}, which
  includes an additional if-condition before adding a cut to check for
  the existence of a suitable cone $\Corner$, allows our B\&C method to
  remain applicable without requiring additional restrictive assumptions
  to guarantee the existence of such a cone.
  \revsecond{The finite-termination guarantee in Section~\ref{sec:fallback}
    concerns integer branching only, it does not cover this additional
    continuous-variable branching.}
\end{remark}

For the NE-free set, we define for all $i \in N$ the set
\begin{align*}
  \freeset_i(x^*,y^*)
  \define
  & \Defset{(x,\proxy) \in \Rall \times
    \R^N}{\proxy_i > \pi_i(y_i^*,x_{-i}), \, y_i^* \in X_i(x_{-i})}
  \\
  =
  & \Defset{(x,\proxy) \in \Rall \times
    \R^N}{ \proxy_i >
    \pi_i(y_i^*,x_{-i}), \, g_i(y_i^*,x_{-i})\leq 0}.
\end{align*}
This set is convex provided that the players' cost functions are
convex in the rivals' strategies and that the set of rivals' strategies
$x_{-i}$ admitting $y_i^*$ as a feasible strategy for player $i$ is
convex, which we formalize in the next lemma.

\begin{lemma}\label{lem: FreeSetConv}
  Let $i \in N(x^*,\proxy^*)$ and
  assume that
  \begin{enumerate}[label=(\roman*)]
  \item the function $ \pi_i(y_i^*,\cdot):\Rallmini \to \R, x_{-i}\mapsto
    \pi_i(y_i^*,x_{-i})$, is convex and \label{lem: FreesetConv: Conc}
  \item $g_i(y_i^*,\cdot):\Rallmini\to \R^{m_i}$,
    $x_{-i} \mapsto g_i(y_i^*,x_{-i})$, is convex.
    \label{lem: FreesetConv: Feas}
  \end{enumerate}
  Then, $\freeset_i(x^*,y^*)$ is a convex set.
\end{lemma}
\begin{proof}
  By rewriting the first condition of $\freeset_i(x^*,y^*)$ via
  $\pi_i(y_i^*,x_{-i})- \proxy_i < 0$ and by using~\ref{lem:
    FreesetConv: Conc}, it follows that
  this is a convex restriction. Since by \ref{lem: FreesetConv: Feas},
  the second condition is convex as well, the convexity of
  $\freeset_i(x^*,y^*)$ follows.
\end{proof}

\rev{Let us note that the assumptions of the last lemma are the reason
   why we need to restrict the class of mixed-integer GNEPs that we
  can tackle to those with cost and constraint functions that are
  convex in the rivals' strategies.}

\begin{lemma}\label{lem: FreeSet}
  We have $(x^*,\proxy^*)\in \freeset_i(x^*,y^*)$ for any
  $i \in N(x^*,\proxy^*)$. Moreover, $\freeset_i(x^*,y^*)$ does not
  contain any point of the intersection $\Nespi\cap C_t\cap B_t$ for all $i
  \in N$.
\end{lemma}
\begin{proof}
  We have $(x^*,\proxy^*) \in \freeset_i(x^*,y^*)$ for any $i
  \in N(x^*,\proxy^*)$ because $y_i^* \in \argmin_{y_i \in
    X_i(x_{-i}^*)} \pi_i(y_i,x_{-i}^*)$ implies
  $\proxy^*_i>\VBR_i(x^*_{-i}) = \pi_i(y_i^*,x^*_{-i})$ and $y^*_i \in
  X_i(x^*_{-i})$.

  Moreover, for any $(\bar{x},\pi(\bar{x}))\in \Nespi\cap
  C_t\cap B_t$ and $i \in N$ with $y^*_i \in X_i(\bar{x}_{-i})$, we
  have that
  \begin{align*}
    \pi_i(\bar{x}) =
    \min_{y_i \in X_i(\bar{x}_{-i})} \pi_i(y_i,\bar{x}_{-i}) \leq
    \pi_i(y_i^*,\bar{x}_{-i})
  \end{align*}
  holds, showing that $(\bar{x},\pi(\bar{x})) \notin
  \freeset_i(x^*,y^*)$.
\end{proof}

The set $\freeset_i(x^*,y^*)$ is, in general, not
suitable for deriving ICs as it is not guaranteed that
$(x^*,\proxy^*)$ belongs to its interior.
This leads us to define, for any $\varepsilon > 0$, an extended
version of $\freeset_i(x^*,y^*)$ via
\begin{align*}
  \freeset[\varepsilon]_i(x^*,y^*) \define
  \Defset{(x,\proxy) \in \Rall \times \R^N}{\proxy_i \geq
  \pi_i(y_i^*,x_{-i}), \, g_i(y_i^*,x_{-i})\leq \varepsilon
  \mathbf{1}},
\end{align*}
where we denote by $\mathbf{1}$ the vector of all ones (in appropriate
dimension).
Provided that no point in $\Nespi\cap C_t\cap B_t$ is contained in
the interior of this extended set, it follows from \Cref{lem:
  FreeSetConv,lem: FreeSet} that  $\freeset[\varepsilon]_i(x^*,y^*)$
is an NE-free set under the assumptions of \Cref{lem: FreeSetConv}.
This naturally raises the question for which values of $\varepsilon >
0$ and under what circumstances can this condition be guaranteed.
In this regard, we provide sufficient conditions in the following.

\revsecond{
\begin{lemma}\label{lem: SuffConNEfree}
  Let $i \in N(x^*,\proxy^*)$ and let the assumptions of
  \Cref{lem: FreeSetConv} hold.
  Assume that for every $\bar{x} \in \Nes$ and
  every $j \in [m_i]$, we have
  $g_{i,j}(y_i^*,\bar{x}_{-i}) \in \Z$ or $g_{i,j}(y_i^*,\bar{x}_{-i}) \leq 0$.
  Then, $\freeset[\varepsilon]_i(x^*,y^*)$ for $\varepsilon = 1$ does not contain any point of
  $\Nespi\cap C_t\cap B_t$ in its interior.
\end{lemma}
\begin{proof}
  Consider a point $(\bar{x},\pi(\bar{x})) \in \Nespi\cap C_t\cap B_t$ as well as
  $\varepsilon =1$.
  Now, assume, for the sake of a contradiction, that the point is in the
  interior of $\freeset[1]_i(x^*,y^*)$.
  By the convexity required in \Cref{lem: FreeSetConv} and
  the defining inequalities of $\freeset[1]_i(x^*,y^*)$ holding strictly
  for $x^*,\proxy^*$, we may infer that $g_{i,j}(y_i^*,\bar{x}_{-i})< {1}$ and
  $\pi_i(\bar{x})> \pi_i(y_i^*,\bar{x}_{-i})$.
  Now for those $j\in[m_i]$ satisfying $g_{i,j}(y_i^*,\bar{x}_{-i}) \in \Z$, it
  follows that $g_{i,j}(y_i^*,\bar{x}_{-i})\leq 0$ and hence we have the latter for all $j\in[m_i]$.
  This shows that $(\bar{x},\pi(\bar{x})) \in \freeset_i(x^*,y^*)$ holds,
  which contradicts \Cref{lem: FreeSet} and thus finishes the proof.
\end{proof}

\begin{remark}
  \begin{thmparts}[label=(\roman*)]
  \item Note that the condition $g_{i,j}(y_i^*,\bar{x}_{-i}) \leq 0$ in the above lemma
    is trivially satisfied for any constraint that does not depend on the rivals' strategies as
    $y_i^*$ is a best-response (\wrt $x^*$) and hence in particular feasible.
  \item \revsecond{Under the assumptions of \Cref{lem: SuffConNEfree},}
    by \Cref{lem: proxy>BR}, we have that $N(x^*,\proxy^*) \neq \emptyset$ holds in the situation of
    Line~\ref{line: proxy>BR}, i.e., in this situation
    there exists $i \in N(x^*,\proxy^*)$ for which
    we can use $\freeset[\varepsilon]_i(x^*,y^*)$ with $\varepsilon = 1$ in
    order to derive a non-NE-cut via an IC.
  \item
    Let us note that analogous assumptions are made in the respective
    literature on mixed-integer bilevel optimization; see, e.g.,
    \rev{\textcite{Wang_Xu:2017},} \textcite{fischetti2018use},
    \textcite{Lozano_Smith:2017}, or \textcite{Horländer_et_al:2024}.
    In practice, very often  the condition $g_{i,j}(y_i^*,\bar{x}_{-i}) \in \Z$ for all $\bar{x}\in \Nes$ in \Cref{lem: SuffConNEfree}
    is satisfied due to constraints being linear, all
    problem data being integer (or, w.l.o.g., rational) and all players'
    variables being integer as well, which is easy to check in
    practice. In particular, any one of the following conditions (for
    a suitable integral matrix~$A_i$ and vector~$b_i$) implies
    $g_{i,j}(y_i^*,\bar{x}_{-i}) \in \Z$ for all $\bar{x}\in \Nes$ and
    all $j\in[m_i]$:
    \begin{enumerate}[label=(\roman*)]
    \item $y_i^*$ is integral and $g_i(y_i^*,\bar{x}_{-i}) = A_i
      (y_i^*,\bar{x}_{-i}^\inte) - b_i$ for all $\bar{x} \in
      \Nes$;\label{lem: SuffConNEfree: IntDep}
    \item $g_i(y_i^*,\bar{x}_{-i}) = A_i
      ((y_i^*)^\inte,\bar{x}_{-i}^\inte) - b_i$ for all $\bar{x} \in
      \Nes$;\label{lem: SuffConNEfree: IntDepOnly}
    \item $y_i^*$ and all  $\bar{x} \in \Nes$ are integral and
      $g_i(y_i^*,\bar{x}_{-i}) = A_i (y_i^*,\bar{x}_{-i})  - b_i$
      holds for all $\bar{x} \in \Nes$.\label{lem: SuffConNEfree:
        IntAll}
    \end{enumerate}
  \end{thmparts}
\end{remark}
}

\revised{Summarizing this section,
  we can construct an NE-free set under the fulfillment of the
  conditions of  Lemma~\ref{lem: SuffConNEfree}.
  Moreover, if we have linear constraints and concave social costs,
  the optimal node solution is attained at a vertex of the polyhedral
  feasible set and we can use the corresponding LP basis cone to
  ensure the existence of a pointed cone $\Corner$ as required for the
  construction of an IC.
  This results in the following theorem.
  \revsecond{}
  \begin{theorem}\label{thm: GNEPs}
    In the situation of Line~\ref{line: Cut} of \Cref{algorithm:BC},
    \revsecond{suppose that the optimal node solution $(x^*,\proxy^*)$
      is chosen to be a vertex of $\FeasRt$.
      Then,} we can guarantee the existence of a non-NE-cut
    via intersection cuts if
    \begin{enumerate}[label=(\roman*)]
    \item the conditions of \Cref{lem: SuffConNEfree} are satisfied, and
    \item the constraints are linear and social costs are concave.
    \end{enumerate}
    In particular, we can  guarantee the existence of a non-NE-cut
    via intersection cuts for pure-integer GNEPs with linear
    constraints and linear player objectives that depend solely on the
    own variables.
  \end{theorem}
}


\revsecond{
\subsection{Fallback Strategies Ensuring Finite Termination for GNEPs}
\label{sec:fallback}

In the previous section, we derived conditions in \Cref{thm: GNEPs}
under which we can derive non-NE-cuts via intersection cuts. It is
well-known that intersection cuts can become arbitrarily shallow even
for single-level mixed-integer optimization problems,
so that B\&C algorithms relying on ICs may fail to terminate in finite time.
In this regard, the processing of a node in our B\&C algorithm via \Cref{algorithm:BC} may not terminate in
finite time either. In the following, we identify a condition for which we can provide a fallback strategy
for the case in which the ICs generated during the execution of \Cref{algorithm:BC} at a node become too shallow.
Roughly speaking, we fix (by branching) the integer components of all players and process
the corresponding node via \Cref{algorithm:BC} with best-response cuts that are only valid in this special scenario.
For these cuts, we can then show that the node is processed in finite time by \Cref{algorithm:BC}
under the assumption that only finitely many sets of best-response sets exist \wrt the fixed integer profile.
Hence,
our overall B\&C algorithm will ultimately terminate in finite time since there are only finitely many possible integer profiles.

In the following, we define for every integer profile $\txint \in \FeasStratsInt$ and $i\in N$ the set
\begin{align*}
  \FeasStrats[-i](\txint) &\define  \Defset{x_{-i} \in \FeasStrats[-i] }{ \xint_{-i} = \txint_{-i}},
\end{align*}
which is the set of rivals' strategies in $\FeasStrats[-i]$ with integer part equal to $\txint_{-i}$.
Moreover, we define analogously to the NEP case
\revsecond{the best-response set \wrt the strategy profile $x \in \FeasStrats$ for GNEPs via
\begin{align*}
  \BR(x) \define \argmin_{y\in X(x)} \sum_{i\in N}\pi_i(y_i,x_{-i}).
\end{align*}
Note that fixing an integer profile of a node does not restrict the integer components
of deviations~$y$.}

\begin{assumption}\label{ass:Fallback}
  For every $\xint \in \FeasStratsInt$,
 the following two statements hold for all $i\in N$:
  \begin{enumerate}[label=(\roman*)]
    \item the constraint sets only depend on $\xint$, i.e.,
      $X_i(\tilde{x}_{-i}) = X_i(x_{-i}')$ for all
      $\tilde{x}_{-i},x_{-i}' \in \FeasStrats[-i](\xint)$;
    \item the set of best-response sets \wrt $\xint$ is finite,
      i.e.~$\abs{\set{\BR(\tilde{x}): \tilde{x}\in \FeasStrats\land
          \txint = \xint}}<\infty$.
  \end{enumerate}
\end{assumption}
\begin{remark}
 \revsecond{Under \Cref{ass:Fallback}(i),
      a proof analogous} to the ones of \Cref{lem:
      FinitelyManyBRSetsConcave,lem: FinitelyManyBRSetsInteger} shows
    that the set of best-response sets \wrt $\xint$ is finite if for
    all $i \in N$, one of the following two statements holds:
    \begin{enumerate}[label=(\roman*)]
    \item  player $i$'s
      \begin{enumerate}
      \item strategy set is polyhedral,
      \item cost function is concave in her continuous variables
        $x_i^\con$, i.e., \revsecond{for all $x_{-i}' \in \FeasStrats[-i](\xint)$ and
        $\hat{x}_i^\inte \in \Z^{k_i}$ for which the domain below is nonempty,} the function
        \revsecond{%
        \begin{align*}
          \pi_i(\hat{x}_i^\inte,\cdot,x_{-i}'):&\Defset{y_i^\con }{ (\hat{x}_i^\inte,y_i^\con) \in X_i(x_{-i}')}
          \to \R,\\
           &x_i^\con \mapsto \pi_i(\hat{x}_i^\inte,x_i^\con,x_{-i}')
        \end{align*}}
        is concave.
      \end{enumerate}
    \item the costs of player $i$ only depend on her own strategy for $x_{-i}'\in \FeasStrats[-i](\xint)$, i.e., there exists a function $\hat{\pi}_i:\R^{k_i+l_i}\to \R$ with $\pi_i(\cdot,x_{-i}') = \hat{\pi}_i$
      for $x_{-i}' \in \FeasStrats[-i](\xint)$.
    \end{enumerate}
\end{remark}

We now describe the fallback strategy in more detail: consider the situation in which the integer strategy profile $\txint$ has been part of
an optimal solution to the node problem for more than $\IntBound$ cut-iterations, for some bound $\IntBound>0$.
\revsecond{Integer branching
creates one child fixing $\xint=\txint$ and siblings covering exactly the
other integer profiles. Nodes in which the integer variables are already fixed directly use the locally valid best-response cuts which we derive in the following.}
While all other nodes are processed as usual,
consider, in the following, an optimal solution $(x^*,\proxy^*)$ to the node problem in which $\xint = \txint$ holds and denote by $\FeasRt$ the current feasible set of the node problem. \revsecond{The following lemma describes the cutting strategy for this node.}

\begin{lemma} \label{lem:fixed-profile}
  Consider the node $t$ in which $\xint = \txint$ holds for \revsecond{every node-feasible $(x,\proxy)\in\FeasRt$.}
  Under \Cref{ass:Fallback}
  the best-response cuts from \Cref{lem: Cut} are applicable, i.e.,
  for every best response $y^*$ to $x^*$,
  the best-response cut given by
  \begin{align*}
    c(x,\proxy;
    x^*, \proxy^*, y^*) \define c_i(x,\proxy;y^*) \define \proxy_i -
    \pi_i(y_i^*,{x}_{-i}) \leq 0
  \end{align*}
  yields a non-NE-cut for every $i \in N(x^*,\proxy^*)$. Moreover,
  a straightforward adaptation of the proof of \Cref{thm: FinitelyManyCutsGeneral} shows that \Cref{algorithm:BC} on Node $t$ terminates after a finite number of steps if we only use the cuts above.
\end{lemma}
\begin{proof}
  Consider an arbitrary equilibrium tuple
  $(\bar{x},\pi(\bar{x})) \in \Nespi\cap \FeasRt$ in the current node.
  By assumption, we have $\bxint = \txint = (x^*)^\inte$, and since \Cref{ass:Fallback}(i) holds, $y^* \in X(x^*)$ implies $y^* \in X(\bar{x})$.
  Consequently,
  the equilibrium condition implies for any $i \in N$ that $\pi_i(\bar{x}) \leq \pi_i(y_i^*,\bar{x}_{-i})$.
  Hence, $c_i(\bar{x},\pi(\bar{x});y^*) \leq 0$
  holds for all $i \in N\supseteq N(x^*,\proxy^*)$.
  In particular, the cut satisfies
  Condition~\ref{def:non-ne-cut:1} of \Cref{def:non-ne-cut}.

  For Condition~\ref{def:non-ne-cut:2}, we obtain, as an immediate
  consequence of $i \in N(x^*,\proxy^*)$, the inequality
  \begin{equation*}
    \proxy_i^* >
    \VBR_i(x^*_{-i}) = \min_{y_i \in X_i(x^*_{-i})}\pi_i(y_i,x_{-i}^*) =
    \pi_i(y_i^*,x_{-i}^*).
  \end{equation*}

  \revsecond{Finally, Condition~\ref{def:non-ne-cut:3} is satisfied since
    $c_i(\cdot,\cdot;y^*)$ is continuous on $\FeasStratsRel\times\R^N$
    by Assumption~\ref{ass:GeneralAss}\ref{ass:GeneralAss:Cont}.}
  \qedhere
\end{proof}

From the above lemma, we get the following theorem:
\begin{theorem}\label{thm:fallback}
  \revsecond{Under \Cref{ass:Fallback}, suppose that a non-NE-cut
  is available and used whenever Line~\ref{line: Cut}
    is reached before using the fallback strategy.
    Then, the B\&C algorithm with the above fallback strategy terminates in finite time.}
\end{theorem}
\revsecond{\begin{proof}
  By \Cref{ass:GeneralAss}, $\FeasStratsRel$ is contained in a bounded box
  with integral bounds. Each integer branch strictly reduces this box in
  every nonempty child. Without branching, infinitely many node iterations would repeat
  some profile more than $\IntBound$ times, triggering fallback. After finitely many branching steps to fix $\xint = \txint$ in a child node, Lemma \ref{lem:fixed-profile} then shows that the fallback strategy processes this node with finitely many cuts. Thus, the B\&C algorithm with the above fallback strategy terminates in finite time.
\end{proof}}
}



\section{Numerical Results}
\label{sec:numerical-results}

In this section, we discuss the numerical results of the methods
presented and analyzed so far.
To this end, we start by discussing some implementation details as
well as the hardware and software setup in Section~\ref{sec:impl-deta}.
We present the different types of games to which we apply
our methods in Section~\ref{sec:test-instances}.
The specific way of generating the test instances is presented in
Section~\ref{sec:generation-of-instances}.
\revsecond{Finally, the actual numerical results are discussed in two
parts.
Section~\ref{sec:analysis-of-the-results} shows the behavior of our
branch-and-cut method on the different test instances, while
Section~\ref{subsec:method_comparisons} provides comparisons with
other methods.}

\subsection{Implementation Details}
\label{sec:impl-deta}

All numerical experiments have been executed on a single core
Intel Xeon Gold 6126 processor at \SI{2.6}{GHz} with \SI{4}{GB} of
RAM. In what follows, all the non-default parameter values have
been chosen based on preliminary numerical testing.

We consider a strategy profile~$x$ to be an NE if $\hat{V}(x) \leq
\revised{10^{-3}}$ holds.
\revised{The implementation of Algorithm~\ref{algorithm:BC} is in
\textsf{C++}\footnote{See
    \url{https://github.com/AloisDuguet/branch-and-cut-for-ipgs} for
the implementation, the instances, and the detailed numerical results.}} and
compiled with \textsf{GCC}~13.3.
For the pruning step in Line~\ref{line: Infeasible}, we check
if the objective value is greater than \revised{$10^{-3}$}.
In addition, in Lines~\ref{line:step_1} and~\ref{line: SolveBR},
we solve MIQPs or MILPs using \textsf{Gurobi}~13.0.1
\parencite{gurobi} with the parameter \textsf{feasTol} set to
$10^{-6}$ and the parameter \textsf{MIPGap} set to its default when
solving the node problem and set to $0$ when solving the best-response
problems. The parameter \textsf{intFeasTol} of \textsf{Gurobi} is set to
$10^{-6}$. Finally, a cut is added in Line \ref{line: Cut} if
$\eta^*_i$ is greater than the best response value by at
least~$10^{-6}$ and if the violation of the cut evaluated at the
optimal solution of the current node $(x^*,\proxy^*)$ is greater than
$5 \cdot 10^{-6}$.

The exploration strategy of the branching scheme is depth-first
search, while the variable chosen for branching is the most
fractional one. In case of a tie, the smallest index is chosen.
While the performance of our method most likely would benefit from
more sophisticated node selection strategies and branching rules,
their study and implementation is out of the scope of this paper.

\revised{The branch-and-prune algorithm of
  \textcite{schwarze2023branch} is implemented in \textsf{MATLAB}~24.2 and run
  from the command line with the batch option.
  Its criterion to check if a feasible strategy profile is an NE is
  modified to fit the one used in Algorithm~\ref{algorithm:BC}.
  For this method, \textsf{Gurobi}~13.0.1 is used to solve MIQPs with the
  parameter \textsf{feasTol} set to its default value $10^{-6}$ and the
  feasibility tolerance used in other places is set to $10^{-5}$.}

\revsecond{The mixed-integer linear bilevel solver
\textsf{MibS}\footnote{available at \url{https://github.com/coin-or/MibS/tree/stable/1.2}}
from \textcite{Tahernejad20MibS} is used to solve the reformulation
as a bilevel problem of the game. The LP solver used inside \textsf{MibS}
is \textsf{CPLEX}~22.2.}

\subsection{Description of the Games for the Numerical Experiments}
\label{sec:test-instances}

\subsubsection{The Knapsack Game}
\label{sec:knapsack-game}

We consider a situation with $n$~players and all of them solve a
knapsack-type problem with $m$ personal items. This game is thus an NEP and we
compute a pure NE. The best-response problem of each player is an MILP
that is given by
\begin{align*}
  \max_{x_i} \quad
  & \sum_{j=1}^m p_{ij} x_{i,j}
    + \sum_{k=1, k\neq i}^n \sum_{j=1}^m C_{ikj} x_{i,j} x_{k,j}
  \\
  \st \quad
  & \sum_{j=1}^m w_{ij} x_{i,j} \leq b_i, \quad x_i \in [0,1]^m, \\
  & x_{i,j} \in \mathbb Z, \quad j \in I \subseteq \{1,\dotsc,m\}.
\end{align*}
Note that the subset~$I$ is the same for all players.
In our numerical study, we consider two scenarios for the subset $ I
\subseteq \{1, \dotsc, m\}$ representing the indivisible items:
the full integer case, where all items are indivisible
($I = \{1, \dotsc, m\}$), and a mixed-integer case,
where only half of the items are subject to integrality constraints.
As usual, we assume that the profits~$p_{ij}$, the weights~$w_{ij}$,
and the capacities $b_i$ are non-negative integers.
For the interaction coefficients $C_{ikj}$ we assume that they are
general integers.
For more details, we refer to \textcite{dragotto2023zero}, where this
problem is considered for the pure integer case.

Note that this game fulfills \Cref{ass:GeneralAss} and, hence,
our B\&C method is applicable. Moreover, it terminates in finite time
using the non-NE-cuts as defined in~\eqref{eq: PlayerCut}. This is
implied by
\Cref{lem: FinitelyManyBRSetsInteger} for the full integer case,
respectively \Cref{lem: FinitelyManyBRSetsConcave} for both cases.

\subsubsection{The Generalized Knapsack Game}
\label{sec:generalized-knapsack-game}

We consider again the situation of $n$~players and all of them solve a
knapsack-type problem with a common set of $m$ items. This game is thus
a GNEP and we look for a pure NE. The best-response
problems are ILPs given (for player~$i$) by
\begin{align*}
  \max_{x_i} \quad
  & \sum_{j=1}^m p_{ij} x_{i,j}
  \\
  \st \quad
  & \sum_{j=1}^m w_{ij} x_{i,j} \leq b_i, \quad x_i \in \set{0,1}^m,
  \\
  & \sum_{k = 1}^n x_{k,j} \leq c_j, \quad j = 1, \dotsc, m,
\end{align*}
with $c_j \in \set{1,\dotsc,n}$.
As before, we assume that the profits~$p_{ij}$, the weights~$w_{ij}$,
and the capacities $b_i$ are non-negative integers.

Note that this game also fulfills \Cref{ass:GeneralAss}.
Moreover, the social cost function $x \mapsto \sum_{i \in N}\pi_i(x)$
is affine in $x$.
As outlined in \Cref{sec:cuts-gne}, this guarantees the existence of
an optimal solution to the node problem at a vertex of the underlying
feasible set.
For such a vertex, we can then use
the associated LP basis cone in order to derive an IC.
Moreover, note that \Cref{lem: FreeSetConv} and
\Cref{lem: SuffConNEfree} are applicable.
For the latter, observe that  all cost and constraint
functions are affine, and since all strategies and all constraint data are
integral, $g_{i,j}(y_i^*,\bar{x}_{-i}) \in \Z$ holds for every~$j$.
Thus, the existence of a suitable NE-free set is guaranteed
and our B\&C method is applicable for this game class.
\revsecond{Finally, note that the above GNEP is a potential
  game. Hence, tailored approaches such as minimizing the potential
  over the set of feasible strategy profiles or running a
  best-response algorithm (see, e.g.,
  \textcite{Sagratella2017AlgorithmsFG})
  are applicable and would most likely outperform our more general
  B\&C method on this class.}

\subsubsection{Implementation Games}

We study a  model of \textcite{Kelly98}  in the domain of TCP-based
congestion control.
To this end, we consider a directed graph~$G=(V,E)$ with nodes~$V$ and
edges~$E$.
The set of players is given by $N= \set{1, \dots, n}$ and each
player~$i\in N$ is associated with an end-to-end pair $(s_i,t_i)\in
V\times V$.
The strategy~$x_i$ of player~$i \in N$ represents an integral
$(s_i, t_i)$-flow with a flow value equal to her demand $d_i \in
\nonNegInts$.
Moreover, a player is restricted in her strategy choice by the
capacity constraints $c \in \nonNegInts^E$, i.e., for given rivals'
strategies $x_{-i}$, her flow $x_i$ has to satisfy the restriction
{$x_i \leq c - \sum_{j\neq i} x_j$}.
Thus, the strategy set of a player $i \in N$ is described by
\begin{equation*}
  X_i(x_{-i}) = X_i' \cap
  \Defset{x_i \in \nonNegInts^E}{x_i \leq c - \sum_{j\neq i} x_j}
  \text{ for all } x_{-i},
\end{equation*}
where $X_i' \define \defset{ x_i\in \Z_{\geq 0}^E}{ A_Gx_i = b_i} \cup
\set{0}$ is the union of the $0$-flow and the flow
polyhedron of player $i$ with $A_G$ being the arc-incidence matrix of
the graph~$G$ and~$b_i$ being the vector with $(b_i)_{s_i} = d_i$,
$(b_i)_{t_i} = -d_i$, and $0$ otherwise.
Note that this allows players to not participate in the game because $x_i
= 0$ is a feasible strategy.
\revsecond{This participation option can be modeled using a binary variable $z_i$
and the linear constraints $A_Gx_i = b_i z_i$ and $0 \leq x_i \leq c z_i$.
The resulting continuous relaxation is polyhedral, as required for the
construction of an LP basis cone in Section~\ref{sec:cuts-gne}.}
All players want to maximize their utility given by $\mu_i^\top x_i$
for player $i$ choosing strategy~$x_i$ for a given vector $\mu_i\in
\R^E_{\geq 0}$.

In addition to the set~$N$ of players, there is a
central authority, which determines a price vector~$p^*\in \R_{\geq 0}^E$
for the edges with the goal to (weakly) \emph{implement} a certain
edge load vector $u \in \R^E_{\geq 0}$, i.e., the authority wants to
determine a price vector $p^*$ such that there exists a strategy
profile~$x^*$ of the players in~$N$ with the following properties.
\begin{enumerate}[label=(\roman*)]
\item \label{impl: 1} The load is at most~$u$, i.e.,
  $\ell(x^*)\define \sum_{i \in N} x_i^* \leq u$.
\item \label{impl: 2} The strategy~$x^*$ is an equilibrium for the given~$p^*$,
  i.e.,
  \begin{equation*}
    x_i^* \in \argmax \Defset{(\mu_i - p^*)^\top x_i}{x_i \in X_i(x_{-i}^*)}
  \end{equation*}
  holds for all $i$.
\item \label{impl: 3} The edges for which the targeted load is not
  fully used have zero price, i.e., $\ell_e(x^*) < u_e$ implies $p_e^* =0$,
\item \label{impl: 4} The price is bounded from above, i.e., $p^*\leq
  p^{\max}$.
\end{enumerate}
Here, $p^{\max}\in \R^E_{\geq 0}$ is some upper bound on the prices
satisfying
\begin{equation*}
  p^{\max}_e > \abs{E} \cdot \max_{e' \in E}(\mu_{i})_{e'} \cdot
  \max_{e' \in E}c_{e'}
\end{equation*}
for all $i \in N$ and $e \in E$.

For the setting in which no capacity constraints are present
and players are allowed to send fractional arbitrary amounts of flow,  \textcite{Kelly98}
proved that every vector $u$ is weakly implementable.
Allowing a fully fractional distribution of the flow, however, is not
possible in some applications---the notion of data packets as
indivisible units seems more realistic.
The issue of completely fractional routing versus integrality
requirements has been explicitly addressed by \textcite{Orda93},
\textcite{HarksK16b}, and \textcite{wang2011}. Recently,
\textcite{HarksSchwarzPricing} introduced a unifying framework for
pricing in non-convex resource allocation games,
which, in particular, encompasses the integrality-constrained version of the model
originally studied by \textcite{Kelly98}. \revsecond{For a related integral-flow model in which the flow amounts are
decision variables, \textcite[Corollary~7.8]{HarksSchwarzPricing} show
that every nonnegative integral capacity vector is weakly implementable
when all players share a common source and have identical nondecreasing
linear bandwidth utilities. This result, however, does not directly apply to the
fixed-demand model considered here.}
However, in the general case, the implementability of a vector $u$ is not guaranteed.
This raises the question of which vectors are implementable and which are not.

We can model this question as a jointly constrained GNEP with $n+1$ players in which the
first $n$~players correspond to the player set~$N$ and the $(n+1)$-th
player is the central authority.
We denote by $(x,p)$ a strategy  profile and set the costs
to the negated utility $\pi_i(x_i,x_{-i},p) = (p-\mu_i)^\top  x_i$
for $i \in N$ and the costs of the central authority
to  $\pi_{n+1}(p,x) = (u - \ell(x))^\top p$.
The joint restriction set $X$ is given by
\begin{equation*}
  X \define \Defset{(x,p)\in \prod_{i \in N}X_i' \times
    \R^E_{\geq 0}}{\ell(x) \leq c, \, p \leq p^{\max}}.
\end{equation*}

Let us make the relation to GNEPs a bit more formal.

\begin{lemma}
  A tuple $(x^*,p^*)$ (weakly) implements $u$ if and only if
  $(x^*,p^*)$ is an equilibrium of the above described GNEP.
\end{lemma}
\begin{proof}
  First, let $(x^*,p^*)$ (weakly) implement~$u$.
  Feasibility, i.e., $(x^*,p^*) \in X$, follows immediately
  from \ref{impl: 2} and \ref{impl: 4}.
  Note that  \ref{impl: 2} in particular implies that $x_i^* \in
  X_i(x_{-i}^*)$ is feasible, leading to $\ell(x^*) \leq c$.
  By \ref{impl: 2}, the players in $N$ also play an optimal strategy.
  It thus remains to verify that $p^*$ is an optimal strategy for
  the central authority. By \ref{impl: 1}, we have $(u- \ell(x^*))\in
  \R_{\geq 0}^E$ and since prices must be nonnegative, the optimal value
  is bounded from below by $0$. By \ref{impl: 3}, we have
  $(u- \ell(x^*))^\top p^* = 0$, showing the claim.

  Let now $(x^*,p^*)$ be an equilibrium of the above described GNEP.
  Condition~\ref{impl: 4} follows immediately by feasibility, i.e.,
  from $(x^*,p^*) \in X$.
  Condition~\ref{impl: 2} follows by the players in $N$ playing an optimal
  strategy.
  Since $p^*$ is an optimal strategy of the central authority,
  \ref{impl: 3} holds as well.
  For \ref{impl: 1}, assume for the sake of a contradiction that there
  exists an edge~$e \in E$ with $\ell_e(x^*) >u_e$.  By optimality of the
  central authority, this implies $p_e^* = p_e^{\max}$.
  Let $i \in N$ be a player with $x_{ie}^*\geq 1$.
  Such a player needs to exist due to $\ell_e(x^*) > u_e \geq 0$ and
  because strategies are required to be integral.
  The utility of player~$i$ is bounded by
  \begin{equation*}
    \mu_i^\top x_i \leq \abs{E} \cdot \max_{e' \in E}(\mu_{i})_{e'} \cdot
    \max_{e' \in E}c_{e'}.
  \end{equation*}
  This together with $x_{ie}^*\geq 1$ leads to the lower bound for
  the costs of player $i$ given by
  \begin{align*}
    \pi_i(x_i^*,x_{-i}^*,p^*) = (p^*-\mu_i)^\top   x_i^* \geq
    p_e^{\max} - \abs{E}\cdot  \max_{e' \in E}(\mu_{i})_{e'} \cdot  \max_{e'
    \in E}c_{e'} > 0,
  \end{align*}
  where the last inequality holds by the definition of $p^{\max}$.
  This, however, contradicts the optimality of player~$i$ since $x_i =
  0$ would lead to zero costs.
\end{proof}

In our branch-and-cut approach, we need to sum up the cost functions
of all players.
In this situation here (including the central authority), these social
costs are given by
\begin{align*}
  \sum_{i \in N} \pi_i(x,p) + \pi_{n+1}(p,x)
  & = \sum_{i \in N} (p-\mu_i)^\top x_i +  (u - \ell(x))^\top p
  \\
  & =
    \sum_{i \in N} -\mu_i^\top x_i  +  \sum_{i \in N} p ^\top
    x_i +  (u - \ell(x))^\top p
  \\
  &=  \sum_{i \in N} -\mu_i^\top x_i  + p ^\top \sum_{i \in N}
    x_i  +  (u - \ell(x))^\top p
  \\
  &= \sum_{i \in N} -\mu_i^\top x_i  + p ^\top \ell(x) +  (u -
    \ell(x))^\top p
  \\
  &= \sum_{i \in N} -\mu_i^\top x_i + u^\top p,
\end{align*}
which is an affine function in $(x,p)$. As outlined in
\Cref{sec:cuts-gne}, this guarantees the existence of an optimal
solution to the node problem at a vertex of the underlying feasibility
set.
For such a vertex, we can then use the associated LP basis cone to
derive an IC.
Moreover, \Cref{lem: FreeSetConv} and \Cref{lem: SuffConNEfree} are applicable:
\revsecond{The costs and constraints of all players are linear in the rivals'
  strategies.
  Furthermore, $\ell(x) - c$ is always integral since $x$ and $c$ are
  integral, while $0 \leq p \leq p^{\max}$ holds for any feasible price
  vector~$p$.
}
Thus, the existence
of a suitable NE-free set is guaranteed. Hence, our B\&C method is
applicable for this class of games.

\begin{remark}
  Let us also brief\/ly note that the game sketched above is not a
  generalized ordinal potential game.
  To show this, consider a graph with two nodes $\set{s,t} = V$, which
  are connected via two parallel edges $\set{e_1,e_2} = E$.
  Let further $N=\set{1}$ be given with the data $d_1 = 2$, $c =
  (3,3)$, $\mu_1 = (2,1)$, and $u = (1,1)$.
  Consider now the following four strategy profiles~$(x,p)$
  together with the corresponding costs $(\pi_1(x,p),\pi_2(p,x))$ for
  any number $M \in (2, p_{e_1}^{\text{max}}]$:
  \begin{align*}
    &x^1=(2,0), p^1=(0,0),(-4,0),
    &\quad
    &x^2=(2,0), p^2=(M,0),(-4+2M,-M),
    \\
    &x^3=(0,2), p^3=(M,0),(-2,M),
    &\quad
    &x^4=(0,2), p^4=(0,0),(-2,0).
  \end{align*}
  Then, the improvement of the player deviating from her strategy in
  the sequence of strategy profiles $(x^i,p^i)$, $i = 1,\ldots,5$, with
  $(x^5,p^5)\define (x^1,p^1)$ is always negative.
  Thus, the game cannot be a generalized ordinal potential game.
\end{remark}

\subsubsection{Integer NEPs with Quadratic Objective Functions}
\label{subsubsection:integerNEP}

\revised{We consider the integer-constrained NEPs used in
  \textcite{schwarze2023branch}:
  \begin{align*}
    \min_{x_i \in \Z^m} \quad
    & \frac{1}{2} x_i^\top Q_i x_i + (C_i x_{-i} + d_i)^\top x_i, \\
    \st \quad
    & A_i x_i \leq b_i, \\
    & -5 \leq x_{i,j} \leq 5, \quad j=1,\ldots,m.
  \end{align*}
  For these games, the objective function can be either convex or
  nonconvex.
  More details about these instances are given below.
  Note that this game fulfills \Cref{ass:GeneralAss} and hence
  our B\&C method is applicable. Moreover, by  \Cref{lem:
    FinitelyManyBRSetsInteger}, it terminates in finite time
  using the non-NE-cuts as defined in~\eqref{eq: PlayerCut}.
}

\subsection{Generation of Instances}
\label{sec:generation-of-instances}

To generate a knapsack game instance, we created $n$ knapsack problems with
the same parameters using Pisinger's knapsack problem generator
described in~\textcite{Silvano99}, where $n$ is the number of players. We
generated instances for $n \in \{2,3,4\}$, number of items
$m \in \{5,10,15,20,30,40,50,60,70,80\}$, and the capacity set to 0.2, 0.5,
or 0.8 times the sum of the weights of items of the respective player.
We also produced instances with different types of correlation between
weights and profits of items: an instance either has them uncorrelated,
weakly correlated, or strongly correlated in the sense of Pisinger's
knapsack problem generator. Finally, we generated 5 instances with
the same parameters. This makes a total of 1350~instances. In
addition, we solve those instances with both all variables being
integer and only variables with even indices being integer to test
Algorithm~\ref{algorithm:BC} on integer and mixed-integer
instances.
We denote those two sets of instances NEP-I and NEP-MI
respectively. In both cases, we use the globally valid non-NE-cuts
described in Equation~\eqref{eq: PlayerCut}.

The GNEP knapsack games are generated in the same way, but with
number of items $m \in \{5,10,15,20,30,40,50\}$, player's capacity
set to 0.2 or 0.5 times the sum of the weights of items of the
respective player, and 10~instances with the same parameters.
We do not use the factor 0.8 here because preliminary results showed
that the resulting instances are too easy.
In addition, the parameter~$c_j$ representing
the amount of item~$j$ available for all players is chosen randomly
and uniformly in $\set{1,\dotsc,n}$. This makes a total of 1260
instances. We consider only instances with all variables integer as
the intersection cuts have been shown to be (locally valid)
non-NE-cuts only for this case.

As for implementation game instances, we generate them with the
instances of the jointly capacitated discrete flow game (JCDFG) of
\textcite{Harks24} in the following way.
The matrix $A_G$ corresponds to matrix $A$ of the
JCDFG, the right-hand side $b_i$ is built from
the sources $s_i$, sinks $t_i$, and the flow demand $d_i$
of the JCDFG, the capacities in the vector $c$ are the ones from the
capacity vector $c$ of the JCDFG, the utility vector $\mu_i$ corresponds
to the vector $C_i^2$ of the linear utility used in the JCDFG, and,
finally, the edge load vector~$u$ is generated in the same way as the
capacities $c$ of the JCDFG.
The non-NE-cuts used in our experiments for this class of problems are
the intersection cuts.

\revised{Finally, the integer NEP instances with quadratic objective
  functions come from the randomly generated benchmark set of
  \textcite{schwarze2023branch}. It consists of 56 instances for which
  the name describes most of the characteristics of the
  instances. The format of the name of an instance is $XAB_k$ where $X
  \in \set{C,N}$ indicates if the instance is player-convex for each
  player ($C$) or not ($N$).
  The two following numbers~$A$ and~$B$ give the number of players and the
  number of variables of each player and finally the index~$k$
  differentiates the instances with the same other
  characteristics.
  The instances have 2 or 3 players, 2 to 5 variables per
  player and each player's optimization problem has 4 to 10
  constraints. For more details about the instances, see
  \textcite{schwarze2023branch}.}

\subsection{Analysis of the Results}
\label{sec:analysis-of-the-results}

Preliminary results showed that the aggregated best-response cuts
described in Equation~\eqref{eq: AggCutOurSetting} perform worse than
the best-response cuts of Equation \eqref{eq: PlayerCut}, so we only show
results for the latter. When we say that an
instance is solved, it means that either an NE has been found or that
we prove non-existence---both within the time limit.

\subsubsection{Knapsack Game Results} \label{subsubsec:knapsackGames}

Approximately \revised{\SI{11}{\%}} of
instances are solved in less than \SI{1}{s}, and
\SI{42}{\%} are solved in the time limit of \SI{1}{h}. All instances
solved have an NE, i.e., non-existence of NE was not proved for any
instance.

\begin{table}
  \centering
  \caption{Number of instances solved by number of players and number of
    items}
  \label{table:NEP_item_player}
  \begin{tabular}{ccccccccccccc} \toprule
    & players & \multicolumn{10}{c}{items} \\ \cmidrule(lr){2-2} \cmidrule(lr){3-12}
    && 5 & 10 & 15 & 20 & 30 & 40 & 50 & 60 & 70 & 80 & \%\\
    \midrule
    &2 & 45&45&45&44&36&26&13&2&6&1 & 58 \\
    NEP-I &3 & 45&45&44&32&8&5&2&1&0&1 & 41 \\
    &4 & 45&45&23&8&1&1&1&0&0&0 & 28 \\ \midrule
    &2 & 45&45&45&42&28&21&11&7&3&2 & 55 \\
    NEP-MI &3 & 45&45&43&29&11&5&1&0&1&0 & 40 \\
    &4 & 45&45&34&16&1&1&0&0&0&0 & 32 \\
    \bottomrule
  \end{tabular}
\end{table}

Table \ref{table:NEP_item_player} shows the number of instances solved
depending on the number of players and the number of items. For each
set of parameters, the algorithm is applied to 45 instances. The last column shows the
percentage of instances solved for the corresponding number of players
among all different number of items.
First, instances with
mixed-integer variables and with only integer variables have
relatively similar results.
Moreover, the results are as expected: the instances get more challenging
when the number of items and the number of players increase.

\begin{figure}
  \includegraphics[width=0.5\linewidth]{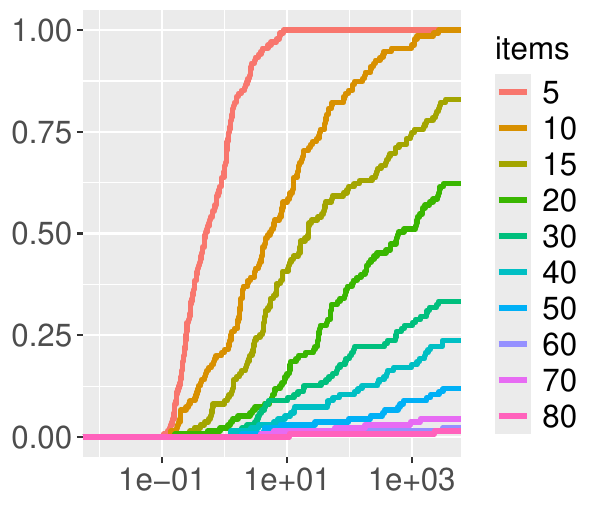}
  \includegraphics[width=0.5\linewidth]{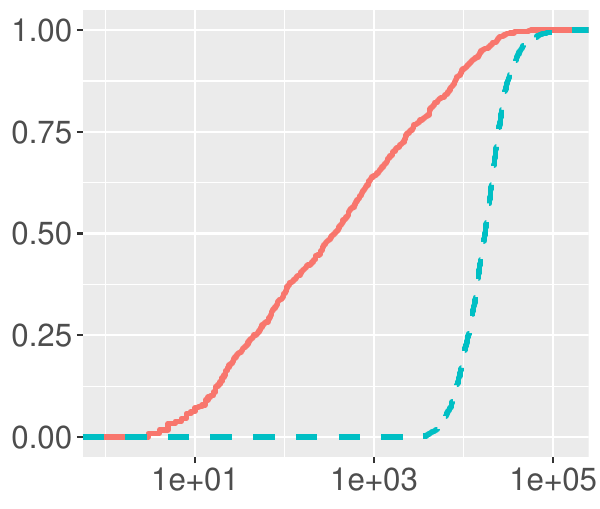}
  \includegraphics[width=0.5\linewidth]{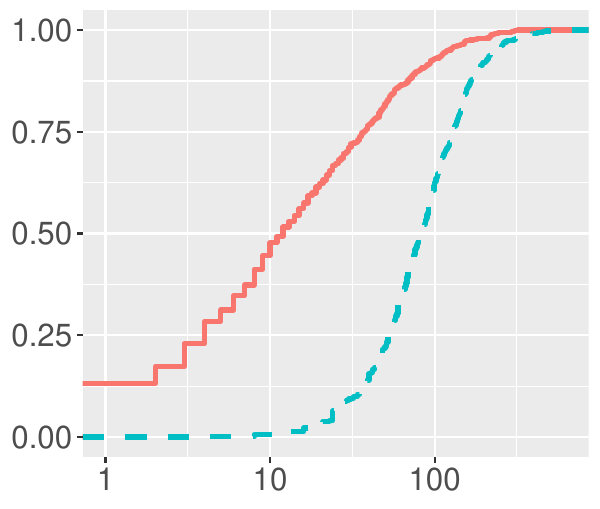}
  \caption{Characteristics of the knapsack game experiments with
    only integer variables. Top: ECDFs of the computation time depending on
    the number of items. Middle: ECDFs of the number of nodes visited in
    the branch-and-cut for solved instances (solid red) and unsolved
    instances (dashed blue). Bottom: ECDFs of the number of cuts derived in
    the branch-and-cut for solved instances (solid red) and unsolved
    instances (dashed blue)}
  \label{figure:ECDF_NEP_knapsack_characteristics}
\end{figure}

The top plot of Figure~\ref{figure:ECDF_NEP_knapsack_characteristics}
shows the empirical cumulative distribution functions (ECDFs) of the
knapsack game instances with integer variables with respect to the
number of items. It can be seen that the difficulty increases rather
regularly with the number of items.
Indeed, instances with 5 items are solved almost instantly, as
\revised{\SI{66}{\%}} of
them are solved in less than \SI{1}{s}, while instances with 50 items
or more are mostly unsolved even after \SI{1}{h} of computation. The
middle plot of Figure~\ref{figure:ECDF_NEP_knapsack_characteristics} shows two ECDFs
of the
number of nodes visited by the branch-and-cut. The solid red curve
considers only solved instances while the dashed blue curve considers
only unsolved instances.
The proportion of instances solved increases
with the number of nodes, so the branching scheme seems to help. Also,
all unsolved instances visited many nodes. In
comparison, the instances with mixed-integer variables produce
much smaller branching trees: they have less than 1400 nodes and
there are unsolved instances with only 3 nodes -- even though the
finite termination of \Cref{algorithm:BC} for a node is guaranteed by
\Cref{lem: FinitelyManyBRSetsConcave}. It thus seems that
the mixed-integer instances are harder to solve. In a similar manner,
the bottom plot of Figure
\ref{figure:ECDF_NEP_knapsack_characteristics}
shows two ECDFs of the number of cuts derived in the branch-and-cut. It
seems that for some instances, very few integral solutions to the node
problem were found. Indeed, there is an instance with 4 cuts
derived that is unsolved.
The number of cuts derived are similar for the mixed-integer variable
instances.

\subsubsection{Generalized Knapsack Game Results}
\label{subsubsec:gnep-knapsack-results}

As a rough overview:
Approximately \revised{\SI{6}{\%}} of instances are solved in less than
\SI{1}{s} and \revised{\SI{34}{\%}} are solved in the time limit of \SI{1}{h}.
All instances solved found an NE.
Table~\ref{table:GNEP_item_player} shows the number of instances
solved depending on the number of players and the number of items. For
each set of parameters, Algorithm \ref{algorithm:BC} is applied to 60
instances. The last column shows the percentage of instances solved
for the corresponding number of players among all different number of
items.
The trends are comparable to the ones we have seen before.
The instances get harder to solve both for an increasing number of
items and for an increasing number of players.

\begin{table}[!hb]
  \centering
  \caption{Number of instances solved by number of players and number of
    items}
  \label{table:GNEP_item_player}
  \begin{tabular}{ccccccccc} \toprule
    players & \multicolumn{7}{c}{items}\\ \cmidrule(lr){1-1} \cmidrule(lr){2-8}
            & 5 & 10 & 15 & 20 & 30 & 40 & 50 & \% \\ \midrule
    2 & 60&60&39&12&8&2&3 & 44 \\
    3 & 60&54&14&13&3&1&0 & 35 \\
    4 & 60&26&7&4&1&2&1 & 24 \\
    \bottomrule
  \end{tabular}
\end{table}

\revsecond{\textcite[Section~6.4 and Appendix~G]{Hao_et_al:2026}
independently reuse these 1,260 instances for their pure-integer algorithm
and report finding an
equilibrium on all of them, with a mean time of $0.066$ seconds in
their single-threaded comparison. Their method performs substantially
better on this pure-integer benchmark, demonstrating the benefit of a specialized approach
for the pure-integer case.
}

\begin{figure}
  \centering
  \includegraphics[width=0.5\textwidth]{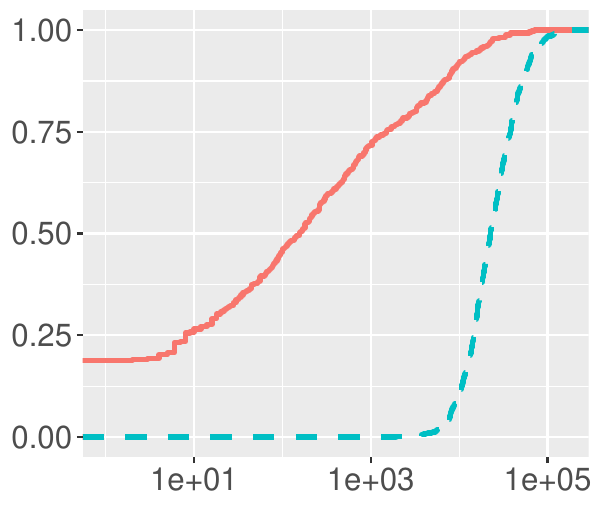}
  \caption{ECDFs of the number of cuts for solved instances (solid red)
    and unsolved instances (dashed blue)}
  \label{figure:GNEP_knapsack_ecdf_cuts}
\end{figure}

The ECDFs of the GNEP knapsack instances with respect to the number of
items and to the number of nodes are similar to the NEP knapsack game
instances with only integer variables. However, as can be seen in
Figure~\ref{figure:GNEP_knapsack_ecdf_cuts}, the number of cuts
derived in the GNEP is way higher than in the NEP case,
with instances with over \num{100000} cuts derived.
This may be due to multiple reasons.
First, the node problem is considerably easier to solve in the GNEP
case, because it is an LP while it is a non-convex QP in the NEP case.
Second, the cuts in the GNEP case are (only) locally valid while they
are globally valid in the NEP case.

\subsubsection{Implementation Game Results}

Approximately \SI{4}{\%} of instances are solved in less than
\SI{1}{s} and \revised{\SI{41}{\%}} are solved in the time limit of \SI{1}{h}.
More precisely, an NE was found for \revised{\SI{40}{\%}} of the instances while
\revised{\SI{1}{\%} (5 out of 450)} were proven to have no
NE. \revsecond{Also, one instance stopped because a large amount of
increasingly shallow ICs has been added at a single node, causing
numerical instabilities.}

\begin{table}
  \centering
  \caption{Percentage of instances solved by number of players and
    number of nodes of the network. All cells contain 40 runs, except two players/20 nodes (130 runs).}
  \label{table:implementationGame_node_player}
  \begin{tabular}{cccc} \toprule
    players & \multicolumn{3}{c}{nodes}\\ \cmidrule(lr){1-1} \cmidrule(lr){2-4}
            & 10 & 15 & 20 \\ \midrule
    2 & 85.0&57.5&55.4 \\
    4 & 60.0&22.5&32.5 \\
    10 & 15.0&10.0&2.5 \\
    \bottomrule
  \end{tabular}
\end{table}

\begin{figure}
  \includegraphics[width=0.5\linewidth]{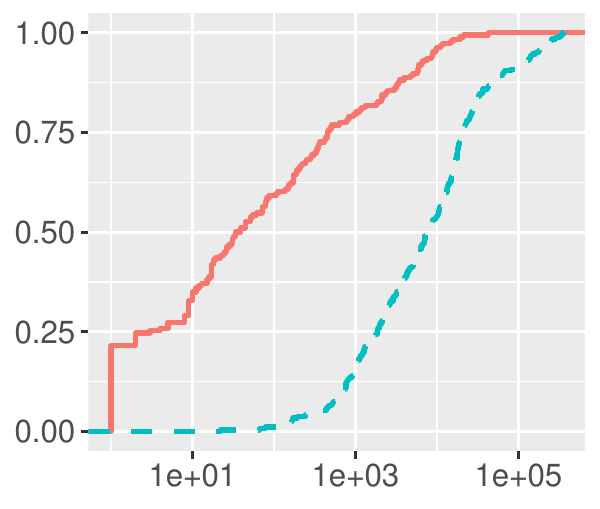}
  \includegraphics[width=0.5\linewidth]{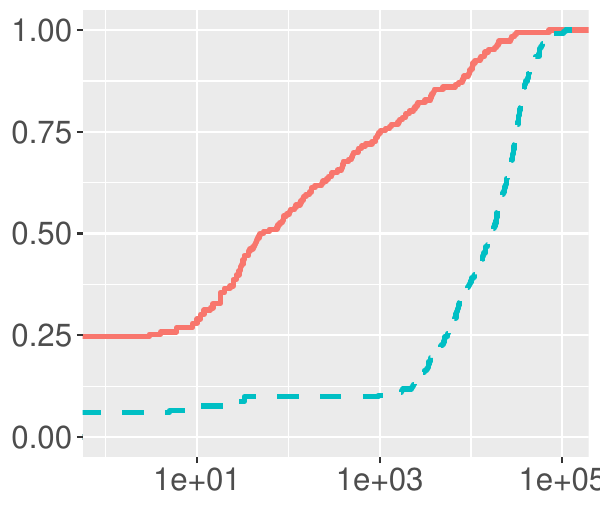}
  \caption{Characteristics of the implementation game experiments. Top:
    ECDFs of the number of nodes visited in the branch-and-cut for
    solved instances (solid red) and unsolved instances (dashed blue).
    Bottom: ECDFs of the number of cuts derived in the branch-and-cut for
    solved instances (solid red) and unsolved instances (dashed blue)}
  \label{figure:ECDF_implementationGame_characteristics}
\end{figure}

Table~\ref{table:implementationGame_node_player} shows the percentage
of instances solved by Algorithm \ref{algorithm:BC}. All instances for
which no NE was proven needed at least 6 cuts.
As can be seen in the top plot of
Figure~\ref{figure:ECDF_implementationGame_characteristics},
\revised{\SI{21}{\%}} of the solved instances are solved in the root node,
i.e., without branching. This might be explained by the fact
that the best responses contain constraints for the flow conservation
in a network. They are formulated with a
totally unimodular matrix and thus the solution of an LP containing only
those constraints would be integral.
In our case, the node problem has the flow conservation constraints as
well as some more constraints, so it might be expected to find an integral
solution fast. Similarly, the bottom plot of
Figure~\ref{figure:ECDF_implementationGame_characteristics} shows that
\revised{\SI{25}{\%}} of solved instances are solved with no cuts derived. In
addition, \revised{\SI{6}{\%}} of the unsolved instances reach the time limit
with no cuts derived.
In this case, there are many visited nodes, which shows that the
instances are large enough so that no integral solution was ever
encountered despite the overall large branching tree.
Such cases should benefit from more involved branching strategies.

\subsubsection{\revised{Time Spent in Different Steps for GNEP Instances}}

\revised{We finally discuss the time spent in the different steps of
  the algorithm when applied to GNEPs.
  The results are different depending on which application we
  consider, i.e., knapsack or implementation games.
  Producing the intersection cuts takes around \SI{22}{\%} of the
  computation time for solving a knapsack instance, but around
  \SI{63}{\%}
  for an implementation game instance. Most of the remaining time is
  spent to solve the node problems.
  Concerning the computation of best responses, it only takes \SI{1}{\%}
  of the computation time for the two applications.
  Finally, most of the computation time to derive an intersection cut
  is spent to compute the extreme rays describing the LP basis cone,
  i.e., around \SI{75}{\%} for the two applications. The remaining time is
  essentially spent to compute the $\alpha$-values.
  Note that more efficient methods for computing ICs are available,
  in particular with respect to the computation of extreme rays. Right now,
  our implementation serves more as a proof-of-concept and
  we leave it for future work to implement improvements that
  reduce the time of
  deriving the intersection cuts to get to an overall more effective
  method.}

\subsection{Comparison to Other Methods}
\label{subsec:method_comparisons}

\revsecond{We compare ourselves to four other methods on relevant instances.
It includes the zero regret algorithm from \textcite{dragotto2023zero}
in Section \ref{subsubsec:dragotto_comparison}
and
the branch-and-prune from \textcite{schwarze2023branch} in Section
\ref{subsubsection:comparisonBPBC}.
It also assesses in Section \ref{subsubsection:bilevel_solver} the
benefit of our tailored branch-and-cut algorithm over
a generic bilevel solver applied to the reformulation of the game, as well
as that of game-specific cuts.
To this end, we compare our method on relevant test instances with a generic
bilevel solver and with a variant of our method that uses no-good cuts, i.e.,
simple cuts that are not game-specific.}

\subsubsection{The Zero-regret Algorithm} \label{subsubsec:dragotto_comparison}
We compare ourselves to the zero regret algorithm from
\textcite{dragotto2023zero} on the instance set of knapsack games from
the same article. For details on the instances, see the original article.

While we solve all instances with 2 players and 25 items within the
time limit, we could only solve 1 instance out of 9 with 2 players and
75 items.
However, the cutting-plane method of \textcite{dragotto2023zero}
solves all instances with 2 players and up to 75 items, as well as 7
instances out of 9 with 100 items.
This result was expected because our approach is more general than
theirs and they implemented two additional types of cuts specific to
the knapsack game while we did not. Note that our method
was able to show the non-existence of NE for two instances in the set.

\subsubsection{\revised{Comparison Between Branch-and-Cut and
    Branch-and-Prune Methods for Integer NEPs}}
\label{subsubsection:comparisonBPBC}


\begin{table}
\caption{Comparison of B\&P and B\&C}
\label{table:comparisonBPBC}
\centering
\footnotesize
\begin{tabular}{cc d{3.2} c d{3.2} rr}
  \toprule
  & \multicolumn{2}{c}{Branch-and-Prune} & \multicolumn{4}{c}{Branch-and-Cut} \\
    \cmidrule{2-3} \cmidrule{4-7}
 ID & status & \multicolumn{1}{c}{\text{time}} & status & \multicolumn{1}{c}{\text{time}} & nodes & cuts \\
  \midrule
$C22_1$ & NE found & 9.7 & NE found & 1.1 &  11 &   0 \\ 
$C22_2$ & NE found & 4 & NE found & 1.4 &  30 &   4 \\ 
$C22_3$ & NE found & 2.7 & NE found & 2 &  52 &   9 \\ 
$C22_4$ & $t_{\text{max}}$ & \multicolumn{1}{c}{---} & no NE & 3.6 &  61 &   9 \\ 
$C23_1$ & NE found & 1341.7 & NE found & 5.3 &  31 &   4 \\ 
$C23_2$ & $t_{\text{max}}$ & \multicolumn{1}{c}{---} & no NE & 663 &  59 &  11 \\ 
$C23_3$ & NE found & 170.8 & NE found & 25.5 &  75 &  11 \\ 
$C23_4$ & $t_{\text{max}}$ & \multicolumn{1}{c}{---} & no NE & 52.5 & 131 &  20 \\ 
$C23_5$ & $t_{\text{max}}$ & \multicolumn{1}{c}{---} & no NE & 19.7 &  97 &  14 \\ 
$C23_6$ & $t_{\text{max}}$ & \multicolumn{1}{c}{---} & NE found & 8.1 &  65 &  10 \\ 
$C23_7$ & NE found & 283.6 & NE found & 7.3 &  27 &   6 \\ 
$C23_8$ & NE found & 86.7 & NE found & 10.1 & 118 &  18 \\ 
$C24_1$ & NE found & 3275.9 & NE found & 250.5 & 118 &  13 \\ 
$C24_2$ & $t_{\text{max}}$ & \multicolumn{1}{c}{---} & no NE & 240.8 & 499 &  36 \\ 
$C24_3$ & $t_{\text{max}}$ & \multicolumn{1}{c}{---} & no NE & 350.6 & 355 &  32 \\ 
$C24_4$ & $t_{\text{max}}$ & \multicolumn{1}{c}{---} & no NE & 595.7 & 421 &  22 \\ 
$C25_1$ & $t_{\text{max}}$ & \multicolumn{1}{c}{---} & NE found & 3057.9 & 393 &  29 \\ 
$C25_2$ & $t_{\text{max}}$ & \multicolumn{1}{c}{---} & $t_{\text{max}}$ & \multicolumn{1}{c}{---} & 463 &  22 \\ 
$C25_3$ & $t_{\text{max}}$ & \multicolumn{1}{c}{---} & NE found & 1423.9 & 268 &  23 \\ 
$C25_4$ & $t_{\text{max}}$ & \multicolumn{1}{c}{---} & $t_{\text{max}}$ & \multicolumn{1}{c}{---} & 141 &  14 \\ 
$C32_1$ & NE found & 70.3 & NE found & 24.4 & 248 &  33 \\ 
$C32_2$ & NE found & 6.2 & NE found & 5.8 &  37 &   9 \\ 
$C32_3$ & NE found & 3.1 & NE found & 35.2 & 403 &  48 \\ 
$C32_4$ & $t_{\text{max}}$ & \multicolumn{1}{c}{---} & no NE & 77 & 193 &  34 \\ 
$C33_1$ & NE found & 388.5 & $t_{\text{max}}$ & \multicolumn{1}{c}{---} & 866 &  78 \\ 
$C33_2$ & $t_{\text{max}}$ & \multicolumn{1}{c}{---} & NE found & 2221.7 & 488 &  52 \\ 
$C33_3$ & $t_{\text{max}}$ & \multicolumn{1}{c}{---} & NE found & 444.7 & 713 &  44 \\ 
$C33_4$ & $t_{\text{max}}$ & \multicolumn{1}{c}{---} & $t_{\text{max}}$ & \multicolumn{1}{c}{---} & 2156 &  67 \\ 
$N22_1$ & NE found & 8.6 & NE found & 3.4 &  18 &   4 \\ 
$N22_2$ & NE found & 68.2 & NE found & 2.4 &  58 &   8 \\ 
$N22_3$ & * & 922.7 & NE found & 5.4 &  67 &  10 \\ 
$N22_4$ & $t_{\text{max}}$ & \multicolumn{1}{c}{---} & no NE & 3.3 &  33 &   8 \\ 
$N23_1$ & $t_{\text{max}}$ & \multicolumn{1}{c}{---} & NE found & 11.1 & 186 &  27 \\ 
$N23_2$ & NE found & 117.2 & NE found & 1.5 &   9 &   4 \\ 
$N23_3$ & $t_{\text{max}}$ & \multicolumn{1}{c}{---} & no NE & 1.6 &  29 &   7 \\ 
$N23_4$ & error & \multicolumn{1}{c}{---} & NE found & 10.3 & 145 &  20 \\ 
$N23_5$ & NE found & 5.7 & NE found & 2.2 &  11 &   4 \\ 
$N23_6$ & $t_{\text{max}}$ & \multicolumn{1}{c}{---} & no NE & 18.5 & 159 &   8 \\ 
$N23_7$ & $t_{\text{max}}$ & \multicolumn{1}{c}{---} & no NE & 5.4 &  21 &   6 \\ 
$N23_8$ & $t_{\text{max}}$ & \multicolumn{1}{c}{---} & no NE & 10 &  79 &  12 \\ 
$N24_1$ & NE found & 5.2 & NE found & 9.7 &  41 &   6 \\ 
$N24_2$ & $t_{\text{max}}$ & \multicolumn{1}{c}{---} & no NE & 92.3 &  21 &   6 \\ 
$N24_3$ & $t_{\text{max}}$ & \multicolumn{1}{c}{---} & no NE & 8.8 &  53 &   8 \\ 
$N24_4$ & $t_{\text{max}}$ & \multicolumn{1}{c}{---} & no NE & 10.5 &  65 &  17 \\ 
$N25_1$ & $t_{\text{max}}$ & \multicolumn{1}{c}{---} & no NE & 77.6 & 125 &  22 \\ 
$N25_2$ & NE found & 14.8 & NE found & 30.3 & 176 &  18 \\ 
$N25_3$ & $t_{\text{max}}$ & \multicolumn{1}{c}{---} & NE found & 24 & 210 &  13 \\ 
$N25_4$ & $t_{\text{max}}$ & \multicolumn{1}{c}{---} & NE found & 23.5 &  87 &  12 \\ 
$N32_1$ & $t_{\text{max}}$ & \multicolumn{1}{c}{---} & no NE & 14.5 & 153 &  12 \\ 
$N32_2$ & NE found & 339.4 & NE found & 2.3 &  14 &   7 \\ 
$N32_3$ & NE found & 36 & NE found & 6.2 &  77 &  13 \\ 
$N32_4$ & NE found & 929.8 & NE found & 4.1 &  41 &   8 \\ 
$N33_1$ & $t_{\text{max}}$ & \multicolumn{1}{c}{---} & NE found & 34.4 & 119 &  31 \\ 
$N33_2$ & $t_{\text{max}}$ & \multicolumn{1}{c}{---} & NE found & 57.5 & 266 &  40 \\ 
$N33_3$ & $t_{\text{max}}$ & \multicolumn{1}{c}{---} & no NE & 55.9 & 131 &  21 \\ 
$N33_4$ & $t_{\text{max}}$ & \multicolumn{1}{c}{---} & NE found & 689.8 & 185 &  29 \\ 
   \bottomrule
\end{tabular}
\end{table}


\revised{We restrict the comparison of our
  method to other single-tree branch-and-bounds,
  namely the B\&P from \textcite{schwarze2023branch}.
  The multi-tree branch-and-bound method of
  \textcite{dragotto2023zero}, even though it can be applied to those
  instances, is thus out of scope. Nevertheless, note that it most likely
  produces better results than the two methods compared below; see the
  discussion in Section \ref{subsubsec:dragotto_comparison}.
  Table~\ref{table:comparisonBPBC} shows the numerical results for the
  benchmark set of \textcite{schwarze2023branch}.
  The column ``ID'' lists
  the instance names, the two columns ``status'' show if the corresponding
  algorithm is able to find an NE (NE found), prove that no NE exists
  (no NE), or do not terminate because of the time limit ($t_{\text{max}}$).
  There are two special cases, each occurring only once.
  In the first case, denoted by an asterisk, the B\&P algorithm produces a provably
  wrong result, because it wrongly prunes a Nash equilibrium.
  In the second case, denoted by ``error'', the internal Gauß--Seidel
  sub-procedure of the method by \textcite{schwarze2023branch} stops
  after a given maximum number of iterations without computing a
  feasible point of the required accuracy.
  The columns ``time'' display the time spent for the corresponding
  algorithm, with a ``---'' indicating that the time limit has been
  reached.
  Note that the B\&P method by \textcite{schwarze2023branch} can
  compute all NE of the given game.
  However, the time displayed in the third column of
  Table~\ref{table:comparisonBPBC} is the one required to compute the
  first NE.
  Finally, columns ``nodes'' and ``cuts'' list the number of nodes explored
  and cuts derived by Algorithm~\ref{algorithm:BC} before stopping.}

\revsecond{The B\&P method solves 21 out of 56 instances, finding an
  equilibrium in those 21 solved instances. On the other hand, the B\&C
  solves 52 instances, among which 19 are proved to have no NEs.
  This suggests that the B\&C is better suited to show
  non-existence than the B\&P.}
  \revised{In addition, we compared the fraction of the computation
  time of the B\&P divided by the computation time of the B\&C for
  instances in which both algorithms find an NE. The geometric mean
  is equal to 6.8, meaning that the B\&C is significantly faster on
  average.}
\revised{Thus the B\&C clearly outperforms the B\&P.}

\subsubsection{Comparison with a Bilevel Solver and with Non-tailored Cuts}
\label{subsubsection:bilevel_solver}
\revsecond{We assess the benefit of our tailored branch-and-cut algorithm over a generic
	bilevel solver applied to the reformulation of the game, as well as the
	benefit of game-specific cuts.}

\revsecond{To do this, we consider three methods. The first method called
\textit{tailored branch-and-cut} is the branch-and-cut handling each
node via Algorithm \ref{algorithm:BC}, which uses the cuts tailored to
games, namely best-response cuts for NEPs and intersection cuts for
GNEPs. The second method called \textit{branch-and-cut with no-good
cuts} is the same branch-and-cut except that the tailored cuts are
replaced with classic no-good cuts. Those cuts can be used for
games for which all variables are binary variables.
They are globally valid and cut exactly one integer point, namely the
integer node solution that is not an NE.
It is used to evaluate the benefit of using tailored cuts.
Finally, the third method is
\textsf{MibS} \parencite{Tahernejad20MibS}, the state-of-the-art
solver for mixed-integer linear bilevel problems. It is used to solve
the bilevel reformulation of the game.
On the one hand, the difference of performance between the tailored
branch-and-cut and the branch-and-cut with no-good cuts shows the
efficiency of the tailored game cuts. On the other hand, the
difference of performance between the tailored branch-and-cut and \textsf{MibS}
shows the benefit of a tailored solver for the bilevel-reformulation of
a game.}

\revsecond{Before discussing the numerical results, let us brief\/ly
  mention that the tailored branch-and-cut
has an advantage in generality over \textsf{MibS}. Indeed, the latter is only
capable of solving bilevel
problems satisfying the following constraints: there are no continuous
linking variables and the constraints, upper-level objective, and
lower-level objective are linear. In particular, it cannot solve the
implementation games or the mixed-integer knapsack games
because they have continuous linking variables. The tailored
branch-and-cut has milder assumptions. Thus, we compare numerically
the three methods on the integer knapsack games and the generalized
knapsack games. The integer knapsack game induces quadratic terms
between binary variables in the bilevel reformulation, which \textsf{MibS}
cannot handle directly, so \textsf{MibS} is called with a linearization of those terms.
The comparison on the integer knapsack games is thus carried out on a
different bilevel reformulation depending on the algorithm.
Table \ref{table:method_comparison} shows the results.
The column \enquote{solved} gives the proportion of instances solved in a time
limit of \SI{3600}{s}, while \enquote{mean time} and \enquote{median time} are the
mean and median of computation time in seconds for the instances
solved by all three methods, which corresponds to 9\% (124 instances) of the integer
knapsack games and 20\% (247 instances) of the generalized knapsack games. The
tailored branch-and-cut has better statistics in all categories
compared to the two other methods.
It shows that the tailored
branch-and-cut performs significantly better than the other two methods.
We conclude that using an algorithm specifically designed for the problem
 structure is beneficial for solving the bilevel problem induced by the
  game, and that the tailored cuts are effective.}

\begin{table}
  \centering
  \caption{Percentage of solved instances and running times for our
    tailored B\&C, B\&C with no-good cuts, and \textsf{MibS}}
  \label{table:method_comparison}
\resizebox{\textwidth}{!}{
  \begin{tabular}{ccccccc} \toprule
    & \multicolumn{3}{c}{integer knapsack} & \multicolumn{3}{c}{generalized knapsack} \\ \cmidrule(lr){2-4} \cmidrule(lr){5-7}
    & solved (\%) & mean time & median time & solved (\%) & mean time & median time \\
    \cmidrule(lr){2-7}
    tailored B\&C & 42 & 0.99 & 0.49 & 34 & 44.2 & 5.99 \\
    B\&C with no-good cuts & 13 & 154 & 6.06 & 26 & 250 & 40.9 \\
    \textsf{MibS} & 11 & 186 & 19.35 & 21 & 200 & 7.34 \\
    \bottomrule
  \end{tabular}}
\end{table}


\section{Conclusion}
\label{sec:conclusion}

In this paper, we presented the first branch-and-cut approach for computing
pure equilibria for mixed-integer games.
We reformulated the problem
as a min-max bilevel problem and transferred techniques from mixed-integer
bilevel optimization to the current setting.
Provided that non-NE-cuts exist and the algorithm terminates,
we proved the correctness of our method in \Cref{thm:correctness_algBC}, stating that
we either find a pure NE of the game or prove their non-existence.
For standard NEPs, we showed that non-NE-cuts always exist via
best-response cuts and proved the finite termination of our method
under the assumption that the set of  best-response sets is finite
(\Cref{thm: FinitelyManyCutsGeneral}).
This condition is then shown to be satisfied for the important special
cases of (i) players' cost functions being concave in their own
continuous strategies with polyhedral strategy sets and (ii) the
players' cost functions only depending on their own strategy and the
rivals' integer strategy components (\Cref{lem:
  FinitelyManyBRSetsConcave,lem: FinitelyManyBRSetsInteger}).

For GNEPs, we identified sufficient conditions
for  the existence of non-NE-cuts
via intersection cuts, which are in particular satisfied
for pure-integer GNEPs with linear constraints and
linear decoupled cost functions (\Cref{thm: GNEPs}).
\revsecond{Moreover, we proposed a fallback strategy that fixes
recurring integer profiles and guarantees finite termination under suitable assumptions (\Cref{sec:fallback}).
\revsecond{As already outlined in the introduction, the rational-input example in
Appendix~\ref{app: Irrational} rules out a general finite-termination guarantee for
our B\&C algorithm using only ICs and vertex node solutions.}
In this regard,}
determining suitable
conditions under which the B\&C algorithm terminates finitely
when employing ICs deserves further investigation and
remains an open direction for future research.

Our numerical results can be seen as a proof of concept and they give
an insight about the instance sizes that can be tackled with the novel
techniques.
\revsecond{On the tested integer and generalized knapsack instances,
our method solves a larger fraction of instances than MibS within
the time limit. The comparison with an otherwise identical variant
using no-good cuts also supports the benefit of the game-specific
cuts.}
Of course, different improvements are possible and interesting future
research directions exist.
For instance, it might be fruitful to derive game-specific node
selection or branching rules as well as to integrate further pruning
techniques from the recent literature.
Finally, we suspect that the usage of intersection cuts can even be
extended to more general convex (instead of linear) cases.
The details are out of the scope of this paper and are left for future work.


\section*{Acknowledgements}

This research has been funded by the Deutsche Forschungsgemeinschaft
(DFG) in the project 543678993 (Aggregative gemischt-ganzzahlige
Gleichgewichtsprobleme: Existenz, Approximation und Algorithmen).
We acknowledge the support of the DFG.
The computations were executed on the high performance cluster
``Elwetritsch'' at the TU Kaiserslautern, which is part of the
``Alliance of High Performance Computing Rheinland-Pfalz'' (AHRP).
We kindly acknowledge the support of RHRK.
Finally, the first and third author thank the DFG for their support
within RTG 2126 ``Algorithmic Optimization''.


\printbibliography

\crefalias{section}{appendix}
\appendix
\markboth{\revsecond{\MakeUppercase{\shortauthors}}}{\revsecond{\shorttitle}}
\section{Technical Auxiliary Results}

\begin{lemma}
  \label{lem: HelpConc}
  Consider a minimization problem
  \begin{equation*}
    \min \quad f(x) \quad \text{s.t.} \quad x \in P
  \end{equation*}
  with $P$ being a polyhedral feasible set and $f$ a concave objective
  function.
  Then, the set of optimal solutions can be described as the union of
  faces of $P$.
\end{lemma}
\begin{proof}
  We require the following statement that every point~$x$ of $P$ is
  contained in the relative interior of a face $\Fa(x)$ of $P$; see,
  e.g., \textcite[Theorem 18.2]{Rockafellar+1970}.

  Denote now by $\Opt$ the set of optimal solutions.
  Then, it holds that if $x\in \Opt$ is an optimal solution and
  if it is contained in the relative interior of a face $\Fa$, then
  $\Fa \subseteq \Opt$.
  To prove this, consider an arbitrary $y \in \Fa \setminus \{x\}$ and
  assume for the sake of a contradiction that $f(y) > f(x)$.
  Since $x$ is in the relative interior of $\Fa$, there exists
  $\lambda < 0$ such that $z \define x + \lambda( y - x) \in \Fa$.
  Since we have $x = \alpha z + (1-\alpha) y$ for $\alpha =
  1 / (1-\lambda) \in (0,1)$, we get by concavity that
  \begin{equation*}
    f(x) \geq \alpha f(z) + (1-\alpha) f(y)
    > \alpha f(x) + (1-\alpha) f(x)
    = f(x),
  \end{equation*}
  and, thus, end up with the desired contradiction.

  With the two claims collected so far, it follows directly that the
  set of optimal solutions is the union of faces via
  \begin{equation*}
    \Opt \subseteq \bigcup_{x\in \Opt} \Fa(x) \subseteq \Opt.
    \qedhere
  \end{equation*}
\end{proof}

\section{An Example Run of the Branch-and-Cut Method}
\label{appendix:BC-example}

\begin{example}
  We illustrate Algorithm~\ref{algorithm:BC} using a knapsack game
  between two players and two items as described in
  Section~\ref{sec:knapsack-game}.
  To stick to minimization in Problem~\eqref{model:CHPR}, the best
  responses of all players are converted to minimization problems.
  The problem of player~1 is given by
  \begin{align*}
    \min_{x_1} \quad
    & - 2 x_{1_1} - 2 x_{1_2} - x_{1_1} x_{2_1} - x_{1_2} x_{2_2}
    \\
    \text{s.t.} \quad
    & x_{1_1} + x_{1_2} \leq 1, \\
    &0 \leq x_{1_j} \leq 1, \quad j = 1,2, \\
    & x_{1_1} \in \nonNegInts
  \end{align*}
  and the problem of player~2 is given by
  \begin{align*}
    \min_{x_2} \quad
    & - 10 x_{2_1} - 9 x_{2_2} -x_{2_1} x_{1_1} - 4 x_{2_2} x_{1_2} \\
    \text{s.t.} \quad
    & 9 x_{2_1} + 8 x_{2_2} \leq 13, \\
    &0 \leq x_{2_j} \leq 1, \quad j=1,2, \\
    & x_{2_1} \in \nonNegInts.
  \end{align*}
  The problem of node 1 as described in~\eqref{model:CHPR} reads
  \begin{align*}
    \min_{(x,\eta)} \quad
    & - 2 x_{1_1} - 2 x_{1_2} - 10 x_{2_1} - 9 x_{2_2} - \eta_1 -
      \eta_2
    \\
    & - x_{1_1} x_{2_1} - x_{1_2} x_{2_2} -x_{2_1} x_{1_1} - 4
      x_{2_2} x_{1_2}
    \\
    \text{s.t.} \quad
    & x_{1_1} + x_{1_2} \leq 1, \\
    & 9 x_{2_1} + 8 x_{2_2} \leq 13, \\
    &-3 \leq \eta_1 \leq 0, \\
    &-\frac{167}{9} \leq \eta_2 \leq 0, \\
    &0 \leq x_{i_j} \leq 1, \quad i=1,2, \ j=1,2.
  \end{align*}
  The B\&C method processes node~1 with Algorithm~\ref{algorithm:BC}.
  It has optimal solution $(x^*,\eta^*)=(0,1,5/9,1,0,0)$ with
  value $-194/9$. The variable $x_{2_1}^*$ is not integer and we
  branch on~$x_{2_1}$.

  Node~2 is then again processed by Algorithm~\ref{algorithm:BC}.
  We consider the root-node problem plus the branching constraint set
  $B_2$ given by $x_{2_1}=1$.
  It has the optimal solution $(x^*,\eta^*)=(0,1,1,1/2,0,0)$
  with value $-19$.
  The solution satisfies all integrality requirements, so the best
  responses with respect to $x^*$ are computed.
  The best response of player~1 to $(x_{2_1}^*,x_{2_2}^*)$ is $(1,0)$
  of value $-3$ and the best response of player~2 to
  $(x_{1_1}^*,x_{1_2}^*)$ is $(1, 1/2)$ of value $-33/2$, so the
  response $x_2^*$ is the best response.
  On the other hand, $x_1^*$ is not a best response to $x_2^*$ so
  $(x_1^*,x_2^*)$ is not an NE.
  Thus, two non-NE-cuts, one for each player, are added to the
  set $C_2$.
  For this example we use best-response cuts, which are shown to be
  non-NE-cuts in Section~\ref{sec:cuts-ne}.
  For player~1, the cut $\eta_1 \leq -x_{2_1} - 2$ is derived.
  \revsecond{It bounds $\eta_1$ by the cost of deviating to $(1,0)$.}
  For player~2, the cut $\eta_2 \leq -x_{1_1} - 2 x_{1_2} - \frac{29}{2}$
  is derived.
  \revsecond{It bounds $\eta_2$ by the cost of deviating to $(1,1/2)$.}
  We can check that the second cut really cuts the optimal solution of
  the node by evaluating the cut at the solution:
  \begin{equation*}
    0 \nleq -\frac{33}{2}.
  \end{equation*}

  This node is not closed as cuts were derived, so we go back to
  Step~\ref{line:step_1}.
  Thus, the algorithm solves the node problem again but now including
  the additional two cuts.
  The solution obtained is $(x^*,\eta^*)=(1,0,1,1/2,-3,-31/2)$ of
  optimal value~$0$.
  This solution satisfies the integrality requirements and the
  computation of best responses leads to $\Psi(x^*, y^*) = 0$.
  It proves that $x^*$ is an NE of the knapsack game, so the algorithm
  returns~$x^*$.
\end{example}

\section{Explicit Construction of an IC}\label{app: ICConstruction}

In the following, we describe the construction of the
intersection cut, given an NE-free set
$\freeset$ (e.g.,~$\freeset[\varepsilon]_i(x^*,y^*)$) and a cone $\Corner$ pointed at
$(x^*,\proxy^*)$ that contains $\Nespi\cap C_t\cap B_t$.
Here, we adapt the procedure in
\textcite[\revsecond{Chapter~6}]{Conforti-et-al:2014} undertaken for the case
of $\Corner$ being the LP basis cone and refer to the latter for
an extensive overview of the topic.

We assume that there exist finitely many extreme rays
\revsecond{represented by nonzero generators}
$\extray_j$, $j \in \Nbas$, that determine $\Corner$, i.e.,
every $(x,\proxy) \in \Corner$ can be described by $(x,\proxy) =
(x^*,\proxy^*) + \sum_{j\in\Nbas} \lambda_j \extray_j$ for
suitably chosen $\lambda_j \in \R_{\geq 0}$.
\revsecond{No normalization of the generators $\extray_j$ is imposed.}
Note that this assumption is satisfied
for~$\Corner$ being the LP basis cone of $(x^*,\proxy^*)$ \wrt
$\FeasRt$; cf.~\textcite{Conforti-et-al:2014}.
For all $j \in \Nbas$, we define
\begin{align}
  \label{def: ICalpha}
  \alpha_j \define \sup \Defset{\alpha \geq 0}{(x^*,\proxy^*) + \alpha
  \extray_j \in \freeset} 
\end{align}
and note that $\alpha_j > 0$ holds since $(x^*,\proxy^*)$ is
contained in the interior of $\freeset$. 

We further assume that the system of inequalities
\begin{align}
  \label{def: ICCut}
  \begin{pmatrix}
    \extray_1^\top\\
    \vdots\\
    \extray_{\abs{\Nbas}}^\top
  \end{pmatrix} a \geq
  \begin{pmatrix}
    1 / \alpha_1\\
    \vdots\\
    1 / \alpha_{\abs{\Nbas}}
  \end{pmatrix}
\end{align}
has a solution, where we set $1 / \alpha_j \define 0$ in case of
$\alpha_j = \infty$.
Note that this system does admit a solution in case of the extreme
rays being linearly independent, which in turn is satisfied for the
LP basis cone; cf.~\textcite{Conforti-et-al:2014}.

\begin{theorem}\label{thm:ICCut}
  Let $\freeset$ 
  be an NE-free set,
  let $\Corner$ be a cone pointed at $(x^*,\proxy^*)$ that contains
  $\Nespi\cap C_t\cap B_t$, and let $a$ be a solution to~\eqref{def:
    ICCut}.
  Then, the intersection cut defined by
  \begin{align}\label{eq: thm:ICCut}
    c_i(x,\proxy;x^*,\proxy^*,y^*)
    \define
    b -	a^\top (x,\proxy) \leq 0
  \end{align}
  with $b \define a^\top (x^*,\proxy^*) + 1$  is a
  non-NE-cut \wrt $(x^*,\proxy^*)$.
\end{theorem}
\begin{proof}
  The point $(x^*,\proxy^*)$ violates the inequality due to
  \begin{align*}
    c_i(x^*,\proxy^*;x^*,\proxy^*,y^*)
    =
    b -	a^\top (x^*,\proxy^*)
    = 1 + a^\top (x^*,\proxy^*) - a^\top (x^*,\proxy^*) = 1 > 0.
  \end{align*}
  Next, we argue that the inequality is valid for any $(x,\pi(x)) \in
  \Nespi\cap C_t\cap B_t$.
  To this end, we prove the following claim.
  \begin{claim}\label{claim:thm:ICCut}
    Let $\inter{\freeset}$ denote the interior of $\freeset$.
    Then, the inclusion
    \begin{align*}
      \Corner \cap \defset{(x,\proxy)}{a^\top (x,\proxy) < b}
      \subseteq \inter{\freeset} 
    \end{align*}
    is valid.
  \end{claim}
  \begin{proofClaim}
    Consider $(x,\proxy) \in \Corner$ with $a^\top (x,\proxy) < b$.
    Then, there exist $\lambda_l\in \R_{\geq 0}$, $l \in{\Nbas}$, with
    $(x,\proxy) = (x^*,\proxy^*) + \sum_{l \in \Nbas} \lambda_l
    \extray_l$.
    By $a^\top (x,\proxy) < b$, we get
    \begin{align*}
      b>a^\top (x,\proxy)
      = a^\top (x^*,\proxy^*) + a^\top\sum_{l \in \Nbas} \lambda_l
      \extray_l.
    \end{align*}
    Since $b = a^\top (x^*,\proxy^*) + 1$, we obtain
    \begin{align}\label{eq: ConvCombHelp}
      1 > a^\top \sum_{l \in \Nbas} \lambda_l \extray_l =
      \sum_{l \in \Nbas} \lambda_l a^\top  \extray_l \geq
      \sum_{l \in \Nbas} \frac{\lambda_l}{\alpha_l}.
    \end{align}

    Let $\Nbas^{<} \define \defset{l \in \Nbas}{\alpha_l < \infty}$ and
    $\Nbas^{\infty} \define \Nbas \setminus \Nbas^{<}$.
    By \eqref{eq: ConvCombHelp}, the point
    \begin{align*}
      z \define (x^*,\proxy^*) + \sum_{l \in \Nbas^{<}} \lambda_l \extray_l
      = \sum_{l \in \Nbas^{<}} \frac{\lambda_l}{\alpha_l}
      \left((x^*,\proxy^*) + \alpha_l\extray_l \right)
      + \Bigl(1 - \sum_{l \in \Nbas^{<}} \frac{\lambda_l}{\alpha_l}\Bigr)
      (x^*,\proxy^*)
    \end{align*}
    is a convex combination of points in $\cl{\freeset}$ and of
    $(x^*,\proxy^*) \in \inter{\freeset}$, the latter with positive weight.
    Hence, $z \in \inter{\freeset}$; see \textcite[Theorem~6.1]{Rockafellar+1970}.
    For $l \in \Nbas^{\infty}$, the half-line
    $(x^*,\proxy^*) + \R_{\geq 0}\,\extray_l$ is contained in $\cl{\freeset}$, so
    $\extray_l$ is a recession direction of $\cl{\freeset}$ and
    $z + \lambda \extray_l \in \inter{\freeset}$ for all $\lambda \geq 0$;
    see \textcite[Theorem~8.3]{Rockafellar+1970}.
    Consequently,
    $(x,\proxy) = z + \sum_{l \in \Nbas^{\infty}} \lambda_l \extray_l
    \in \inter{\freeset}$.
  \end{proofClaim}
  Since $\Nespi \cap C_t \cap B_t \subseteq \Corner$ and $\Nespi \cap
  C_t \cap B_t \cap \inter{\freeset} =
  \emptyset$
  by $\freeset$ being an NE-free
  set, it follows from the above claim that the inequality is valid for any
  $(x,\pi(x)) \in \Nespi\cap C_t\cap B_t$.
\end{proof}

\begin{remark}\label{rem: CompAlph}
  Assume in \Cref{def: ICalpha} that $\alpha_j< \infty$,  $\freeset =
  \freeset[\varepsilon]_i(x^*,y^*)$, and
  \begin{equation*}
    g_i(y_i^*,\cdot):\prod_{j \neq i}\R^{k_j+l_j}\to \R^{m_i}
  \end{equation*}
  is lower-semicontinuous. \revsecond{Assume also that $\pi_i(y_i^*,\cdot)$ is lower semicontinuous on
    $\prod_{j \neq i}\R^{k_j+l_j}$.}
  Then,
  \begin{align}
    \label{def: ICalphaClosure}
    \alpha_j = \max \Defset{\alpha \geq 0}{(x^*,\proxy^*) + \alpha
    \extray_j \in \freeset[\varepsilon]_i(x^*,y^*)}
  \end{align}
  holds.
  In particular, \revsecond{if both $\pi_i(y_i^*,\cdot)$ and $g_i(y_i^*,\cdot)$ are affine},
  then $\freeset[\varepsilon]_i(x^*,y^*)$ is a polyhedron and the
  optimization problem in \Cref{def: ICalphaClosure} becomes an LP in
  one variable.
\end{remark}

\revsecond{The following example gives a stronger root objective bound with an
IC than with a no-good cut.}
\begin{example}
  Consider a game with 2 items and 3 players $N = \{1,2,3\}$ in which
  all players have to  choose one of the two items.
  Item~1 is only available twice, resulting in the jointly
  constrained GNEP, i.e., $X_i(x_{-i})=\defset{x_i}{(x_i,x_{-i})\in
    X}$, with feasible strategy set
  \begin{equation*}
    X \define \Defset{x \in \Z^6_{\geq 0}}{\sum_{i\in N}x_{i1}\leq 2,
      \, x_{i1} + x_{i2} = 1, \, i \in N}.
  \end{equation*}
  Suppose that the players' costs are given by $\pi_1(x_1,x_{-1})
  = 1 - x_{11}$ and $\pi_i(x_i,x_{-i}) = 1 - x_{i2} + x_{11}$ for $i
  \in \set{2,3}$.

  As an upper bound for $\proxy$ in \Cref{model:node_t}, we use
  \(\proxy^+ \define (2,2,2)\), as no player can incur costs
  exceeding~2.
  The social optimum is given by $x_i^* = (0,1)$ for $i \in N$.
  Together with $\proxy^* = (2,2,2) = \proxy^+$, this is the
  optimal solution in the root node of the B\&C tree.
  The corresponding best response is given by $y^*_1 \define (1,0)$
  and $y_i^*=(0,1)$ for $i \in \set{2, 3}$ with $\VBR_i(x^*) = 0$ for
  $i \in \set{1,2,3}$.
  The resulting IC (\wrt $\Corner$ being the corresponding LP basis
  cone and $\freeset = \freeset[1]_1(x^*,y^*)$) is then
  given by
  \begin{align*}
    \proxy_1 \leq x_{21}+x_{31}
  \end{align*}
  \revsecond{giving root relaxation values $-3$ with this IC and $-4$ with the
    no-good cut
    \begin{align*}
      \sum_{i\in N} x_{i1} \geq 1.
    \end{align*}
    Neither cut dominates the other as a feasible-set inequality.}
\end{example}

We close this section with a discussion about the set
$S_i(x^*,y^*)$ and its corresponding IC in the standard NEP setting,
i.e., if $g_i(x) = g_i(x_i)$ only depends on player $i$'s own variables.
\revsecond{We retain the assumption of \Cref{lem: FreeSetConv} that the
function $x_{-i}\mapsto \pi_i(y_i^*,x_{-i})$ is convex.}
In this case, we have
\begin{equation}
  \label{eq: EQCut}
  \begin{split}
    S_i(x^*,y^*)
    &= \Defset{(x,\proxy) \in \Rall \times \R^N}{ \proxy_i >
      \pi_i(y_i^*,x_{-i})}
    \\
    & = \Defset{(x,\proxy) \in \Rall \times
      \R^N}{(x,\proxy) \text{ violates } \eqref{eq: PlayerCut}}
  \end{split}
\end{equation}
for any $i \in N(x^*,\proxy^*)$ and $S_i(x^*,y^*)$ is in fact a
suitable NE-free set as -- in contrast to the general GNEP
case -- $(x^*,\proxy^*)$ belongs to the interior of the
latter set due to $i \in N(x^*,\proxy^*)$.
For the IC derived in \Cref{thm:ICCut} \wrt $S_i(x^*,y^*)$, we get
as an immediate consequence of \Cref{claim:thm:ICCut} and \eqref{eq:
  EQCut} that any point that is cut off by the IC, i.e., contained
in
\begin{equation*}
  \FeasRt \cap \defset{(x,\proxy)}{a^\top (x,\proxy) < b} \subseteq
  \Corner \cap \defset{(x,\proxy)}{a^\top (x,\proxy) < b},
\end{equation*}
is also cut off by the best-response cut, i.e., contained in
$S_i(x^*,y^*)$.

As a special case, if the cost function $\pi_i(y_i^*,x_{-i})$ is
\revsecond{linear in $x_{-i}$ and all $\alpha_j$, $j\in\Nbas$, are finite}, then there exists a vector $a$ solving \eqref{def:
  ICCut} such that the corresponding IC defined in \Cref{thm:ICCut}  is
equivalent to the best-response cut $\proxy_i \leq \pi_i(y_i^*,x_{-i})$
in \eqref{eq: PlayerCut}. The properties of the IC in the NEP case are
formally stated in the following lemma.

\begin{lemma}\label{lem: ICCutStandNEP}
  For a standard NEP \revsecond{with $x_{-i}\mapsto
    \pi_i(y_i^*,x_{-i})$ being convex}, the IC of \Cref{thm:ICCut} for $\freeset =
  S_i(x^*,y^*)$ is at most as strong as the best-response cut defined in
  \Cref{lem: Cut}, i.e., any $(x,\proxy) \in \Corner$ satisfying \eqref{eq: PlayerCut}  also satisfies \eqref{eq:
    thm:ICCut}. Moreover, if
  $x_{-i}\mapsto \pi_i(y_i^*,x_{-i})$ is linear and all $\alpha_j,j \in
  \Nbas$ in~\eqref{def: ICalpha} are bounded, then there exists a vector
  $a$ solving \eqref{def: ICCut} such that the corresponding IC defined
  in \Cref{thm:ICCut}  is equivalent to the best-response cut~\eqref{eq:
    PlayerCut}, i.e.,~$(x,\proxy)\in\Corner$ satisfies \eqref{eq: thm:ICCut} if and
  only if it  satisfies \eqref{eq: PlayerCut} in this case.
\end{lemma}
\begin{proof}
  As outlined above, the first statement of the \namecref{lem:
    ICCutStandNEP} follows immediately by \Cref{claim:thm:ICCut} and
  \eqref{eq: EQCut}.

  Now assume that $x_{-i}\mapsto \pi_i(y_i^*,x_{-i})$ is linear and
  all $\alpha_j, j \in \Nbas$, in~\eqref{def: ICalpha} are bounded.
  Denote by $\pi_i(y^*_i)$ the vector such that $\pi_i(y_i^*,x_{-i}) =
  \pi_i(y^*_i)^\top x_{-i}$ for all $x_{-i} \in \FeasStrats[-i]$
  and by
  $\extray_j^{x} = (\extray_j^{x_{i}}, \extray_j^{x_{-i}}),
  \extray_j^{\proxy} = (\extray_j^{\proxy,i})_{i  \in N}$ the partial
  vectors of $\extray_j$ such that $\extray_j = (\extray_j^{x},
  \extray_j^{\proxy})$.
  Now for any $j \in \Nbas$, we have
  \begin{align*}
    & \proxy^*_i + \alpha_j\extray_j^{\proxy,i} - \pi_i(y^*_i)^\top
      \left(x^*_{-i} +\alpha_j\extray_j^{x_{-i}}\right) = 0
    \\
    \iff
    & \extray_j^{\proxy,i}-\pi_i(y^*_i)^\top\extray_j^{x_{-i}} =
      \left(\pi_i(y^*_i)^\top x^*_{-i}- \proxy_i^*\right)
      \frac{1}{\alpha_j}
    \\
    \iff
    & \frac{1}{\pi_i(y^*_i)^\top x^*_{-i}- \proxy_i^*}
      \left(\extray_j^{\proxy,i}-\pi_i(y^*_i)^\top\extray_j^{x_{-i}}\right) =
      \frac{1}{\alpha_j}.
  \end{align*}

  We now define $a$ as the vector with
  \begin{align*}
    a^\top (x,\proxy) = \frac{1}{\pi_i(y^*_i)^\top x^*_{-i}-
    \proxy_i^*}
    \left( \proxy_i - \pi_i(y_i^*)^\top x_{-i} \right)
    \text{ for all } (x,\proxy) \in \Rall\times \R^N,
  \end{align*}
  which then implies that for any $\extray_j$, it holds
  \begin{align*}
    a^\top \extray_j
    = \frac{1}{\pi_i(y^*_i)^\top x^*_{-i} - \proxy_i^*}
    \left( \extray_j^{\proxy,i} - \pi_i(y^*_i)^\top
    \extray_j^{x_{-i}}\right)
    = \frac{1}{\alpha_j}.
  \end{align*}
  Hence, $a$ is a solution to~\eqref{def: ICCut} and the corresponding IC
  in \Cref{thm:ICCut} is exactly the best-response cut.
  To see this, note that $b = a^\top (x^*,\proxy^*) +1 = -1 +1 = 0$
  in~\Cref{thm:ICCut} and, hence, the IC is
  \begin{align*}
    & -a^\top (x,\proxy) \leq 0
    \\
    \iff
    & - \frac{1}{\pi_i(y^*_i)^\top x^*_{-i}- \proxy_i^*}
      \left(\proxy_i - \pi_i(y_i^*)^\top x_{-i} \right) \leq 0
    \\
    \iff
    & \proxy_i - \pi_i(y_i^*)^\top x_{-i} \leq 0,
  \end{align*}
  where we used that $\pi_i(y^*_i)^\top x^*_{-i}- \proxy_i^*<0$ holds
  due to $i \in N(x^*,\proxy^*)$.
\end{proof}

\revsecond{
\section{A GNEP with a Single Irrational Equilibrium}
\label{app: Irrational}

In this section, we adapt the example in \parencite{duguet2026mixedintegergamesadmitrational}
such that the resulting GNEP fits into the framework.
\begin{example}\label{ex: IrrationalLinSC}
  We consider a game with three players.
  Players~$1$ and~$2$ each have one integer variable $\xint_i$ and one
  continuous variable $\xcon_i$, and their strategy sets are given by
  \begin{align*}
    X_1(x_{-1})
    & = \Defset{(\xint_1,\xcon_1) \in \set{0,1} \times \R_+}
      {\xcon_1 \leq \xint_1,\ \xcon_1 \leq \xint_2,\
      \xcon_1 + \xcon_2 \leq 1},
    \\
    X_2(x_{-2})
    & = \Defset{(\xint_2,\xcon_2) \in \set{0,1} \times \R_+}
      {\xcon_2 \leq \xint_2,\ \xcon_2 \leq \xint_1,\
      \xcon_1 + \xcon_2 \leq 1},
  \end{align*}
  i.e., players~$1$ and~$2$ share a capacity of~$1$, which a player can
  only use if both players are active.
  Player~$3$ has a single integer variable with strategy set
  $X_3(x_{-3}) = \set{0,1}$.
  We write a strategy profile as
  $x = ((\xint_1,\xcon_1),(\xint_2,\xcon_2),x_3)$.
  The cost functions are given by
  \begin{align*}
    \pi_1(x)
    & = (-\xcon_2 + 3\xint_2 + 3x_3 - 3)\xcon_1
      + \Big(\frac{3}{2} - \xcon_2 - \xint_2 - 2x_3\Big)\xint_1,
    \\
    \pi_2(x)
    & = (\xcon_1 + \xint_1 + 3x_3 - 3)\xcon_2
      + \Big(\frac{3}{2} - 3\xcon_1 + \xint_1 - 2x_3\Big)\xint_2,
    \\
    \pi_3(x)
    & = \Big(-\frac{1}{2} + 2\xint_1 + 2\xint_2 - 3\xcon_1
      - 3\xcon_2\Big)x_3.
  \end{align*}
  All data are rational and all constraints are linear.
  Moreover, the social cost is linear.
  To see this, note that the bilinear terms $\xcon_1\xcon_2$,
  $\xcon_1\xint_2$, $\xint_1\xcon_2$, and $\xint_1\xint_2$ appear in
  $\pi_1$ with the coefficients $-1$, $3$, $-1$, and $-1$ and in $\pi_2$
  with the coefficients $1$, $-3$, $1$, and $1$, respectively.
  Further, the bilinear terms $\xcon_1x_3$ and $\xint_1x_3$ appear in
  $\pi_1$ with the coefficients $3$ and $-2$ and in $\pi_3$ with the
  coefficients $-3$ and $2$, respectively, and the bilinear terms
  $\xcon_2x_3$ and $\xint_2x_3$ appear in $\pi_2$ with the coefficients
  $3$ and $-2$ and in $\pi_3$ with the coefficients $-3$ and $2$,
  respectively.
  Hence, all bilinear terms cancel out and we obtain
  \begin{align*}
    \sum_{i \in N} \pi_i(x)
    = -3\xcon_1 + \frac{3}{2}\xint_1 - 3\xcon_2 + \frac{3}{2}\xint_2
    - \frac{1}{2}x_3.
  \end{align*}
  In the following, we show that $\Nes = \set{\hat{x}}$ holds for the
  irrational strategy profile
  \begin{align*}
    \hat{x} \define
    \left(\left(1,\frac{\sqrt{2}}{2}\right),
    \left(1,1-\frac{\sqrt{2}}{2}\right),0\right)
  \end{align*}
  and that the assumptions of \Cref{lem: SuffConNEfree} are satisfied
  for every node solution.
  To describe the best responses of players~$1$ and~$2$, we define the
  function
  \begin{align*}
    h \colon \R \to \R, \quad r \mapsto r^2 - \frac{1}{2}.
  \end{align*}
  Note that $h(r) \leq 0$ holds for $r \geq 0$ if and only if
  $r \leq \sqrt{2}/2$.

  \begin{claim}\label{claim: IrrLinSC-BR1}
    Let $x_{-1}$ be a rivals' strategy profile with $\xint_2 = 1$,
    $\xcon_2 \in [0,1]$, and $x_3 = 0$.
    Then, every best response to $x_{-1}$ is given by $(0,0)$ or by
    $(1,1-\xcon_2)$, where the latter is a best response if and only
    if $h(1-\xcon_2) \leq 0$ holds.
    Moreover, we have
    \begin{align*}
      \pi_1((1,1-\xcon_2),x_{-1}) = h(1-\xcon_2)
      \quad \text{and} \quad
      \VBR_1(x_{-1}) = \min\set{0,h(1-\xcon_2)}.
    \end{align*}
  \end{claim}
  \begin{proofClaim}
    For $\xint_1 = 0$, the constraints $0 \leq \xcon_1 \leq \xint_1$
    imply $\xcon_1 = 0$.
    For $\xint_1 = 1$, the constraints of player~$1$ read
    $0 \leq \xcon_1 \leq 1$, $\xcon_1 \leq \xint_2 = 1$, and
    $\xcon_1 \leq 1 - \xcon_2$, where the last one is the tightest upper
    bound because of $\xcon_2 \geq 0$.
    Hence, we have
    \begin{align*}
      X_1(x_{-1})
      = \set{(0,0)} \cup \Defset{(1,\xcon_1)}
      {\xcon_1 \in [0,1-\xcon_2]}.
    \end{align*}
    The strategy $(0,0)$ yields the costs $\pi_1((0,0),x_{-1}) = 0$.
    For $\xint_1 = 1$, inserting $\xint_2 = 1$ and $x_3 = 0$ yields
    \begin{align*}
      \pi_1((1,\xcon_1),x_{-1})
      = -\xcon_2\xcon_1 + \frac{1}{2} - \xcon_2.
    \end{align*}
    If $\xcon_2 = 0$ holds, all strategies $(1,\xcon_1)$ yield the costs
    $\tfrac{1}{2} = h(1) > 0$, so that $(0,0)$ is the unique best
    response and the claim holds.
    If $\xcon_2 > 0$ holds, the above expression is strictly decreasing
    in $\xcon_1$.
    Thus, $(1,1-\xcon_2)$ is the unique optimal strategy among all
    strategies with $\xint_1 = 1$ and its costs are given by
    \begin{align*}
      \pi_1((1,1-\xcon_2),x_{-1})
      & = -\xcon_2(1 - \xcon_2) + \frac{1}{2} - \xcon_2
      \\
      & = (1 - \xcon_2)^2 - \frac{1}{2}
        = h(1-\xcon_2).
    \end{align*}
    Comparing these costs with the costs~$0$ of the strategy $(0,0)$
    shows the claim.
  \end{proofClaim}

  \begin{claim}\label{claim: IrrLinSC-BR2}
    Let $x_{-2}$ be a rivals' strategy profile with $\xint_1 = 1$,
    $\xcon_1 \in [0,1]$, and $x_3 = 0$.
    Then, every best response to $x_{-2}$ is given by $(0,0)$ or by
    $(1,1-\xcon_1)$, where the latter is a best response if and only
    if $-h(\xcon_1) \leq 0$ holds.
    Moreover, we have
    \begin{align*}
      \pi_2((1,1-\xcon_1),x_{-2}) = -h(\xcon_1)
      \quad \text{and} \quad
      \VBR_2(x_{-2}) = \min\set{0,-h(\xcon_1)}.
    \end{align*}
  \end{claim}
  \begin{proofClaim}
    As in the proof of \Cref{claim: IrrLinSC-BR1}, the constraints
    $\xcon_2 \leq \xint_1 = 1$ and $\xcon_2 \leq 1 - \xcon_1$ yield
    \begin{align*}
      X_2(x_{-2})
      = \set{(0,0)} \cup \Defset{(1,\xcon_2)}
      {\xcon_2 \in [0,1-\xcon_1]}
    \end{align*}
    and we have $\pi_2((0,0),x_{-2}) = 0$.
    For $\xint_2 = 1$, inserting $\xint_1 = 1$ and $x_3 = 0$ yields
    \begin{align*}
      \pi_2((1,\xcon_2),x_{-2})
      = (\xcon_1 - 2)\xcon_2 + \frac{5}{2} - 3\xcon_1.
    \end{align*}
    Since the coefficient $\xcon_1 - 2 \leq -1$ of $\xcon_2$ is
    negative, $(1,1-\xcon_1)$ is the unique optimal strategy among all
    strategies with $\xint_2 = 1$ and its costs are given by
    \begin{align*}
      \pi_2((1,1-\xcon_1),x_{-2})
      & = (\xcon_1 - 2)(1 - \xcon_1) + \frac{5}{2} - 3\xcon_1
      \\
      & = -(\xcon_1)^2 + 3\xcon_1 - 2 + \frac{5}{2} - 3\xcon_1
        = -h(\xcon_1).
    \end{align*}
    Comparing these costs with the costs~$0$ of the strategy $(0,0)$
    shows the claim.
  \end{proofClaim}

  \begin{claim}\label{claim: IrrLinSC-Int}
    Every NE $\bar{x} \in \Nes$ satisfies
    $\bar{x}_1^\inte = \bar{x}_2^\inte = 1$ and $\bar{x}_3 = 0$.
  \end{claim}
  \begin{proofClaim}
    We write $\pi_3(x) = c_3(x)x_3$ with
    $c_3(x) \define -\tfrac{1}{2} + 2\xint_1 + 2\xint_2 - 3\xcon_1
    - 3\xcon_2$.
    Hence, the strategy $x_3 = 0$ is a best response to $x_{-3}$ if and
    only if $c_3(x) \geq 0$, and the strategy $x_3 = 1$ is a best
    response to $x_{-3}$ if and only if $c_3(x) \leq 0$.
    Let $\bar{x} \in \Nes$.
    We distinguish two cases.
    \begin{description}
    \item[Case 1: $\bar{x}_3 = 1$]
      We first show $\bar{x}_1^\con = 0$.
      If $\bar{x}_1^\inte = 0$ or $\bar{x}_2^\inte = 0$ holds, this
      follows from the constraints $\xcon_1 \leq \xint_1$ and
      $\xcon_1 \leq \xint_2$.
      Otherwise, inserting $\bar{x}_2^\inte = 1$ and $\bar{x}_3 = 1$
      shows that the coefficient of $\xcon_1$ in
      $\pi_1((1,\xcon_1),\bar{x}_{-1})$ is given by
      $-\bar{x}_2^\con + 3 + 3 - 3 = 3 - \bar{x}_2^\con \geq 2$.
      Thus, $\bar{x}_1^\con > 0$ would imply that the feasible strategy
      $(1,0)$ yields strictly smaller costs than $\bar{x}_1$, which
      contradicts $\bar{x}$ being an NE.
      Analogously, we obtain $\bar{x}_2^\con = 0$, where we use that the
      coefficient of $\xcon_2$ in $\pi_2((1,\xcon_2),\bar{x}_{-2})$ is
      given by $\bar{x}_1^\con + 1 + 3 - 3 = \bar{x}_1^\con + 1 \geq 1$
      if $\bar{x}_1^\inte = 1$ holds.
      Consequently, the optimality of player~$3$ implies
      $c_3(\bar{x}) = -\tfrac{1}{2} + 2\bar{x}_1^\inte + 2\bar{x}_2^\inte
      \leq 0$, i.e., $\bar{x}_1^\inte = \bar{x}_2^\inte = 0$.
      In particular, we have $\pi_1(\bar{x}) = 0$.
      However, the strategy $(1,0)$ is feasible for player~$1$ because
      of $0 \leq 1$, $0 \leq \bar{x}_2^\inte$, and
      $0 + \bar{x}_2^\con \leq 1$, and it yields the costs
      \begin{align*}
        \pi_1((1,0),\bar{x}_{-1})
        = \Big(\frac{3}{2} - 0 - 0 - 2\Big) \cdot 1
        = -\frac{1}{2} < 0 = \pi_1(\bar{x}),
      \end{align*}
      which contradicts $\bar{x}$ being an NE.
    \item[Case 2: $\bar{x}_3 = 0$]
      The optimality of player~$3$ implies $c_3(\bar{x}) \geq 0$, i.e.,
      \begin{align*}
        2\bar{x}_1^\inte + 2\bar{x}_2^\inte
        \geq \frac{1}{2} + 3\bar{x}_1^\con + 3\bar{x}_2^\con
        \geq \frac{1}{2}.
      \end{align*}
      Hence, at least one of $\bar{x}_1^\inte$ and $\bar{x}_2^\inte$
      equals~$1$. We argue in the following that both equal $1$.
      Assume, for the sake of a contradiction, that
      $\bar{x}_1^\inte = 1$ and $\bar{x}_2^\inte = 0$ hold.
      Then, the constraints $\xcon_1 \leq \xint_2$ and
      $\xcon_2 \leq \xint_2$ imply
      $\bar{x}_1^\con = \bar{x}_2^\con = 0$, so that
      \begin{align*}
        \pi_1(\bar{x})
        = \Big(\frac{3}{2} - 0 - 0 - 0\Big) \cdot 1
        = \frac{3}{2} > 0 = \pi_1((0,0),\bar{x}_{-1}),
      \end{align*}
      which contradicts $\bar{x}$ being an NE.
      Analogously, $\bar{x}_1^\inte = 0$ and $\bar{x}_2^\inte = 1$ imply
      $\bar{x}_1^\con = \bar{x}_2^\con = 0$ and
      $\pi_2(\bar{x}) = \tfrac{3}{2} > 0 = \pi_2((0,0),\bar{x}_{-2})$,
      which is a contradiction as well.
      Thus, we have $\bar{x}_1^\inte = \bar{x}_2^\inte = 1$.
    \end{description}
    Since Case~1 is impossible, Case~2 shows the claim.
  \end{proofClaim}

  \begin{claim}\label{claim: IrrLinSC-Unique}
    Every NE $\bar{x} \in \Nes$ satisfies $\bar{x} = \hat{x}$.
  \end{claim}
  \begin{proofClaim}
    Let $\bar{x} \in \Nes$.
    By \Cref{claim: IrrLinSC-Int}, we have
    $\bar{x}_1^\inte = \bar{x}_2^\inte = 1$ and $\bar{x}_3 = 0$.
    Moreover, the constraints $0 \leq \xcon_i \leq \xint_i$ yield
    $\bar{x}_1^\con, \bar{x}_2^\con \in [0,1]$.
    Hence, \Cref{claim: IrrLinSC-BR1} applies to $\bar{x}_{-1}$.
    Since $\bar{x}_1$ is a best response to $\bar{x}_{-1}$ with
    $\bar{x}_1^\inte = 1$, we obtain $\bar{x}_1 = (1,1-\bar{x}_2^\con)$
    and $h(1-\bar{x}_2^\con) \leq 0$.
    Analogously, \Cref{claim: IrrLinSC-BR2} applies to $\bar{x}_{-2}$
    and yields $\bar{x}_2 = (1,1-\bar{x}_1^\con)$ and
    $-h(\bar{x}_1^\con) \leq 0$.
    Using $1 - \bar{x}_2^\con = \bar{x}_1^\con$, both inequalities
    together yield $h(\bar{x}_1^\con) = 0$, i.e.,
    $(\bar{x}_1^\con)^2 = \tfrac{1}{2}$.
    Since $\bar{x}_1^\con \geq 0$ holds, we obtain
    $\bar{x}_1^\con = \sqrt{2}/2$ and
    $\bar{x}_2^\con = 1 - \bar{x}_1^\con = 1 - \sqrt{2}/2$, i.e.,
    $\bar{x} = \hat{x}$.
  \end{proofClaim}

  \begin{claim}\label{claim: IrrLinSC-Exist}
    The strategy profile $\hat{x}$ is an NE.
  \end{claim}
  \begin{proofClaim}
    We first show that $\hat{x}$ is feasible.
    We have $\hat{x}_1^\con = \sqrt{2}/2 \in [0,1]$ and
    $\hat{x}_2^\con = 1 - \sqrt{2}/2 \in [0,1]$.
    Together with $\hat{x}_1^\inte = \hat{x}_2^\inte = 1$, this implies
    $\hat{x}_i^\con \leq \hat{x}_1^\inte$ and
    $\hat{x}_i^\con \leq \hat{x}_2^\inte$ for $i \in \set{1,2}$.
    Moreover, we have $\hat{x}_1^\con + \hat{x}_2^\con = 1$.
    Hence, $\hat{x}_i \in X_i(\hat{x}_{-i})$ holds for all
    $i \in \set{1,2,3}$.

    It remains to show that every player plays a best response.
    The rivals' strategy profile $\hat{x}_{-1}$ satisfies the
    assumptions of \Cref{claim: IrrLinSC-BR1}, and we have
    $\hat{x}_1 = (1,1-\hat{x}_2^\con)$ as well as
    $h(1-\hat{x}_2^\con) = h(\sqrt{2}/2) = 0$.
    Thus, $\hat{x}_1$ is a best response to $\hat{x}_{-1}$.
    Analogously, \Cref{claim: IrrLinSC-BR2} applies to $\hat{x}_{-2}$,
    and $\hat{x}_2 = (1,1-\hat{x}_1^\con)$ as well as
    $-h(\hat{x}_1^\con) = -h(\sqrt{2}/2) = 0$ show that $\hat{x}_2$ is a
    best response to $\hat{x}_{-2}$.
    Finally, player~$3$ plays a best response because of
    $c_3(\hat{x}) = -\tfrac{1}{2} + 2 + 2 - 3(\hat{x}_1^\con +
    \hat{x}_2^\con) = \tfrac{1}{2} > 0$.
  \end{proofClaim}

  By \Cref{claim: IrrLinSC-Unique} and \ref{claim: IrrLinSC-Exist}, we have
  $\Nes = \set{\hat{x}}$, i.e., the game only admits an irrational
  equilibrium.
  We finally show that the assumptions of \Cref{lem: SuffConNEfree}
  hold for every node solution.

  \begin{claim}\label{claim: IrrLinSC-SuffCon}
    Let $x \in \FeasStrats$, $i \in \set{1,2,3}$, and let $y_i$ be a
    best response to $x_{-i}$.
    Then, we have $g_{i,j}(y_i,\hat{x}_{-i}) \leq 0$ for all
    $j \in [m_i]$.
  \end{claim}
  \begin{proofClaim}
    For every constraint $g_{i,j}(x_i,x_{-i}) \leq 0$ that does not
    depend on $x_{-i}$, we have
    $g_{i,j}(y_i,\hat{x}_{-i}) = g_{i,j}(y_i,x_{-i}) \leq 0$ due to
    $y_i \in X_i(x_{-i})$.
    In particular, this covers all constraints of player~$3$.
    The remaining constraints of player~$1$ are
    $\xcon_1 - \xint_2 \leq 0$ and $\xcon_1 + \xcon_2 - 1 \leq 0$.
    Inserting $\hat{x}_{-1}$ into the first one yields
    $y_1^\con - 1 \leq 0$, which holds due to
    $y_1^\con \leq y_1^\inte \leq 1$.
    Inserting $\hat{x}_{-1}$ into the second one yields
    $y_1^\con + \hat{x}_2^\con - 1 = y_1^\con - \sqrt{2}/2$.
    Thus, it suffices to show $y_1^\con \leq \sqrt{2}/2$.
    Since $x \in \FeasStrats$ holds, we have
    $x_2^\inte, x_3 \in \set{0,1}$ and $x_2^\con \in [0,x_2^\inte]$, and
    we distinguish three cases.
    \begin{description}
    \item[Case 1: $x_2^\inte = 0$]
      The constraint $\xcon_1 \leq \xint_2$ yields $y_1^\con = 0$.
    \item[Case 2: $x_2^\inte = 1$ and $x_3 = 1$]
      If $y_1^\inte = 0$ holds, we have $y_1^\con = 0$.
      Otherwise, the coefficient of $\xcon_1$ in
      $\pi_1((1,\xcon_1),x_{-1})$ is given by
      $-x_2^\con + 3 + 3 - 3 = 3 - x_2^\con \geq 2$.
      Hence, $y_1^\con > 0$ would imply that the feasible strategy
      $(1,0)$ yields strictly smaller costs than $y_1$, which
      contradicts $y_1$ being a best response.
      Thus, $y_1^\con = 0$ holds.
    \item[Case 3: $x_2^\inte = 1$ and $x_3 = 0$]
      By \Cref{claim: IrrLinSC-BR1}, we have $y_1 = (0,0)$ or
      $y_1 = (1,1-x_2^\con)$ with $h(1-x_2^\con) \leq 0$.
      In the latter case, $y_1^\con = 1 - x_2^\con \geq 0$ and
      $h(y_1^\con) \leq 0$ imply $y_1^\con \leq \sqrt{2}/2$.
    \end{description}
    The remaining constraints of player~$2$ are
    $\xcon_2 - \xint_1 \leq 0$ and $\xcon_1 + \xcon_2 - 1 \leq 0$.
    As before, the first one yields $y_2^\con - 1 \leq 0$ at
    $\hat{x}_{-2}$, while the second one yields
    $\hat{x}_1^\con + y_2^\con - 1 = y_2^\con - (1 - \sqrt{2}/2)$.
    Thus, it suffices to show $y_2^\con \leq 1 - \sqrt{2}/2$, and we
    again distinguish three cases.
    \begin{description}
    \item[Case 1: $x_1^\inte = 0$]
      The constraint $\xcon_2 \leq \xint_1$ yields $y_2^\con = 0$.
    \item[Case 2: $x_1^\inte = 1$ and $x_3 = 1$]
      If $y_2^\inte = 0$ holds, we have $y_2^\con = 0$.
      Otherwise, the coefficient of $\xcon_2$ in
      $\pi_2((1,\xcon_2),x_{-2})$ is given by
      \begin{align*}
        x_1^\con + 1 + 3 - 3 = x_1^\con + 1 \geq 1.
      \end{align*}
      As in Case~2 for player~$1$, we obtain $y_2^\con = 0$.
    \item[Case 3: $x_1^\inte = 1$ and $x_3 = 0$]
      By \Cref{claim: IrrLinSC-BR2}, we have $y_2 = (0,0)$ or
      $y_2 = (1,1-x_1^\con)$ with $-h(x_1^\con) \leq 0$.
      In the latter case, $x_1^\con \geq 0$ and $h(x_1^\con) \geq 0$
      imply $x_1^\con \geq \sqrt{2}/2$ and thus
      $y_2^\con = 1 - x_1^\con \leq 1 - \sqrt{2}/2$.
    \end{description}
    This shows the claim.
  \end{proofClaim}

  Since best responses are only computed for node solutions
  $x^* \in \FeasStrats$ and since $\Nes = \set{\hat{x}}$ holds,
  \Cref{claim: IrrLinSC-SuffCon} shows that, for every such node
  solution $(x^*,\proxy^*)$, every best response~$y^*$ to~$x^*$, and
  every $i \in N(x^*,\proxy^*)$, we have
  $g_{i,j}(y_i^*,\bar{x}_{-i}) \leq 0$ for all $\bar{x} \in \Nes$ and
  all $j \in [m_i]$.
  Moreover, the functions $x_{-i} \mapsto \pi_i(y_i,x_{-i})$ and
  $x_{-i} \mapsto g_i(y_i,x_{-i})$ are affine for every strategy~$y_i$
  of player~$i$, so that the assumptions of \Cref{lem: FreeSetConv}
  are satisfied as well.
  Hence, \Cref{lem: SuffConNEfree} is applicable at every node
  solution, although the game only admits an irrational equilibrium.
\end{example}

\revsecond{Consider now finally
	our B\&C method based on ICs  selecting
  vertex node solutions. Suppose we have rational upper bounds and
  rational basic best responses, and that we solve the simplicial ray
  system with equality.
  Moreover, assume we use rational, not necessarily normalized,
  generators for the basis rays.
  Then, the best-response slices and NE-free sets are rational
  polyhedra at rational node profiles.
  Thus, these choices preserve rational node data, and every
  selected strategy profile is rational and differs from $\hat{x}$.
  \Cref{thm:correctness_algBC} excludes a non-existence certificate because
  $\hat{x}$ exists. Hence this variant cannot terminate.}
}


\end{document}